\documentclass{lmcs} %
\pdfoutput=1
\usepackage[utf8]{inputenc}

\usepackage{lastpage}
\lmcsdoi{21}{4}{30}
\lmcsheading{}{\pageref{LastPage}}{}{}%
{Nov.~08,~2024}{Dec.~15,~2025}{}

\def\eg{{e.g.}}
\def\ie{{i.e.}}

\usepackage{xspace}
\usepackage{amsmath}
\usepackage{amssymb}
\usepackage{centernot}
\usepackage{colonequals}
\usepackage{bm}
\usepackage{subcaption}
\usepackage{thmtools}
\usepackage{thm-restate}

\newcommand{\edge}[2]{#1 \xrightarrow{} #2}
\newcommand{\ledge}[3]{#1 \xrightarrow{#2} #3}
\newcommand{\redge}[3]{#1 \xleftarrow{#2} #3}

\newcommand{\prov}{\mathsf{Prov}}
\newcommand{\minl}{\mathsf{IF}}
\newcommand{\rpq}{\mathsf{RPQ}}
\newcommand{\pqe}[1]{\textsc{PQE}_{#1}}
\newcommand{\mods}{\mathsf{Mods}}

\newcommand{\homo}{\rightsquigarrow}

\newcommand{\poly}{\mathsf{poly}\xspace}
\newcommand{\NP}{\mathsf{NP}\xspace}
\newcommand{\RP}{\mathsf{RP}\xspace}
\newcommand{\PTIME}{\mathsf{PTIME}\xspace}

\newcommand{\all}{\ensuremath{\mathsf{All}}\xspace}

\newcommand{\dagi}{\ensuremath{\mathsf{DAG}}\xspace}
\newcommand{\Bdagi}{\ensuremath{\bm{\mathsf{DAG}}}\xspace}

\newcommand{\phom}{\ensuremath{\mathsf{PHom}}\xspace}
\newcommand{\phoml}{\ensuremath{\mathsf{PHom_{L}}}\xspace}
\newcommand{\phomul}{\ensuremath{\mathsf{PHom_{\centernot L}}}\xspace}

\newcommand{\provl}{\prov}

\newcommand{\owp}{\ensuremath{\mathsf{1WP}}\xspace}
\newcommand{\Bowp}{\ensuremath{\bm{\mathsf{1WP}}}\xspace}
\newcommand{\twp}{\ensuremath{\mathsf{2WP}}\xspace}
\newcommand{\Btwp}{\ensuremath{\bm{\mathsf{2WP}}}\xspace}
\newcommand{\pt}{\ensuremath{\mathsf{PT}}\xspace}
\newcommand{\Bpt}{\ensuremath{\bm{\mathsf{PT}}}\xspace}
\newcommand{\dwt}{\ensuremath{\mathsf{DWT}}\xspace}
\newcommand{\Bdwt}{\ensuremath{\bm{\mathsf{DWT}}}\xspace}

\newcommand{\Gmc}{\ensuremath{\mathcal{G}}\xspace}
\newcommand{\Hmc}{\ensuremath{\mathcal{H}}\xspace}

\newcommand{\countppdnf}{\ensuremath{\text{\#}\textsc{PP2DNF}}\xspace}

\newcommand{\mc}{\mathsf{MC}}
\newcommand{\wmc}{\mathsf{WMC}}

\newcommand\crule[3][black]{\textcolor{#1}{\rule{#2}{#3}}}

\newcommand{\gsize}[1]{||#1||}

\newcommand{\SPhardness}{\mbox{$\#\mathsf{P}$-hardness}\xspace}
\newcommand{\SPhard}{\mbox{$\#\mathsf{P}$-hard}\xspace}
\newcommand{\SP}{\mbox{$\#\mathsf{P}$}}
\newcommand{\NPhard}{\mbox{$\NP$-hard}\xspace}

\usepackage{colortbl}

\theoremstyle{plain}
\newtheorem{result}[thm]{Result}

\usepackage{hyperref}
\hypersetup{hypertexnames=false}

\begin{document}

\title[Approximating Queries on Probabilistic Graphs]{Approximating Queries on 
Probabilistic Graphs\rsuper*}
\titlecomment{{\lsuper*}This is an extended version of an ICDT'24 conference
paper~\cite{AvBM24}. It fixes the proof of Proposition~4.1 and other minor
errors, and includes new material 
on RPQs from Gaspard's bachelor thesis~\cite{G24}.}
\author[A.~Amarilli]{Antoine Amarilli\lmcsorcid{0000-0002-7977-4441}}[a]
\author[T.~van Bremen]{Timothy van Bremen\lmcsorcid{0009-0004-0538-3044}}[b]
\author[O.~Gaspard]{Octave Gaspard\lmcsorcid{0009-0000-7375-3538}}[c]
\author[K.~S.~Meel]{Kuldeep S.~Meel\lmcsorcid{0000-0001-9423-5270}}[d]

\address{Univ.\ Lille, Inria Lille, CNRS, Centrale Lille, UMR
9189 CRIStAL, F-59000 Lille, France;\newline LTCI, Télécom Paris, Institut polytechnique de Paris}	%
\email{antoine.a.amarilli@inria.fr}  %

\address{Nanyang Technological University}	%
\email{timothy.vanbremen@ntu.edu.sg}  %

\address{\'{E}cole Polytechnique}	%
\email{octave.gaspard@polytechnique.edu}  %

\address{University of Toronto}	%
\email{meel@cs.toronto.edu}  %

\thanks{This project was supported in part by the National Research Foundation Singapore under its NRF Fellowship programme [NRF-NRFFAI1-2019-0004] and Campus for Research Excellence and Technological Enterprise (CREATE) programme, as well as the Ministry of Education Singapore Tier 1 and 2 grants R-252-000-B59-114 and MOE-T2EP20121-0011.
	Amarilli
	was partially supported by the ANR project EQUUS ANR-19-CE48-0019, by the 
	Deutsche Forschungsgemeinschaft (DFG, German Research Foundation)~– 431183758,
	and by the ANR project ANR-18-CE23-0003-02 (“CQFD”).
	This work was done in part while Amarilli was visiting the Simons Institute
	for the Theory of Computing.}

\begin{abstract}
	\noindent Query evaluation over probabilistic databases is notoriously intractable---not only in combined complexity, but often in data complexity as well. This motivates the study of
        approximation algorithms, and particularly of \textit{combined FPRASes}, with runtime polynomial in both the query and instance size. In this paper, we focus on tuple-independent probabilistic databases over binary signatures, i.e., \textit{probabilistic graphs}, and study when we can
	devise combined FPRASes for probabilistic query evaluation.

        We settle the complexity of this problem for a variety of query and instance classes, by proving both approximability results and (conditional) inapproximability results together with (unconditional) DNNF provenance circuit size lower bounds. This allows us to deduce many corollaries of possible independent interest. For example, we show how the results of Arenas et al.~\cite{ACJR21} on counting fixed-length strings accepted by an NFA imply the existence of an FPRAS for the two-terminal network reliability problem on directed acyclic graphs, a question asked by Zenklusen and Laumanns \cite{ZL11}. We also show that one cannot extend a recent result of van Bremen and Meel~\cite{vBM23} giving a combined FPRAS for self-join-free conjunctive queries of bounded hypertree width on probabilistic databases: neither the bounded-hypertree-width condition nor the self-join-freeness hypothesis can be relaxed. 
        We last show how our methods can give insights on the evaluation and approximability of \emph{regular path queries} (RPQs) on probabilistic graphs in the data complexity perspective, showing in particular that some of them are (conditionally) inapproximable.
\end{abstract}%

\maketitle

\section{Introduction}%
\textit{Tuple-independent probabilistic databases} (TID)
are a simple and principled formalism 
to model uncertainty and noise in relational
data~\cite{DBLP:conf/vldb/DalviS04,DBLP:series/synthesis/2011Suciu,DBLP:journals/jacm/DalviS12}.
In the TID model, each tuple of a relational database is annotated with an
independent probability of existence; all tuples are assumed to be independent.
In the \textit{probabilistic query evaluation} (PQE) problem, given a Boolean query~$Q$ and a TID instance~$I$, 
we must compute the probability that $Q$ holds in a subinstance sampled from~$I$ according to the resulting distribution.
The PQE problem has been studied in database theory both in terms of \emph{combined complexity}, where the query and instance are part of the input, and in \emph{data complexity}, where the query is fixed and only the instance is given as input~\cite{DBLP:conf/stoc/Vardi82}.
Unfortunately, many of the results so far~\cite{DBLP:series/synthesis/2011Suciu} show that the PQE problem is highly
intractable, even in data complexity for many natural queries (\eg, a path query of length three), and hence also in combined complexity.

Faced with this intractability, 
a natural approach is to study \emph{approximate PQE}: we relax the
requirement of computing the exact probability that the query holds, and settle for an approximate answer. This approach has been studied in data complexity~\cite{DBLP:series/synthesis/2011Suciu}: for any fixed union of conjunctive queries (UCQ), we can always tractably approximate the answer to PQE,
additively (simply by Monte Carlo sampling),
or multiplicatively (using the Karp-Luby approximation
algorithm on a disjunctive-normal-form
representation of the query provenance).
However, these approaches are
not tractable in combined complexity, and moreover the latter approach exhibits a ``slicewise polynomial'' runtime of the form $O(|I|^{|Q|})$---rather than, say, $O(2^{|Q|}\poly(|I|))$---which seriously limits its practical utility. Thus, our goal is to obtain
a \textit{combined FPRAS} for PQE:
by this we mean a fully polynomial-time randomized approximation scheme, giving a multiplicative approximation of the probability, whose runtime is polynomial in the query and TID (and in the desired precision).
This approach has been recently proposed by van Bremen and Meel~\cite{vBM23},
who show a combined FPRAS for CQs when assuming that the query is
self-join-free and has bounded hypertree width; their work leaves open the
question of which other cases admit combined FPRASes.

\paragraph*{Main Results.}
In this paper, following the work of Amarilli, Monet and Senellart~\cite{AMS17} for exact PQE, we investigate the combined complexity of \textit{approximate} PQE in the setting of \textit{probabilistic graphs}. In other words, we study \textit{probabilistic graph homomorphism}, which is the equivalent analogue of CQ evaluation: given a (deterministic) query graph $G$, and given a instance graph $H$ with edges annotated with independent probabilities (like a TID), we wish to approximate the probability that a randomly selected subgraph $H' \subseteq H$ admits a homomorphism from~$G$.
This setting is incomparable to that of~\cite{vBM23}, because it allows for self-joins and for queries of unbounded width, but assumes that relations are binary.

Of course, the graph homomorphism problem is intractable in combined
complexity if the input graphs are arbitrary (even without probabilities).
Hence, we study the problem when the query graph and instance graph are required
to fall in restricted graph classes, chosen to ensure 
tractability in the non-probabilistic setting. We use similar classes as those
from~\cite{AMS17}: \emph{path graphs} which may be \emph{one-way} (\owp: all edges are
oriented from left to right) or \emph{two-way} (\twp: edge orientations are
arbitrary); \emph{tree graphs} which may be \emph{downward} (\dwt: all edges
are oriented from the root to the leaves) or \emph{polytrees} (\pt: edge
orientations are arbitrary); and, for the instance graph, \emph{directed acyclic
graphs} (\dagi), or \emph{arbitrary graphs} (\all).

For all combinations of these classes, we show either (i) the
existence of a combined FPRAS, or (ii) the non-existence of such an FPRAS,
subject to standard complexity-theoretic assumptions. We summarize our results
in Table~\ref{tab:tab}, respectively for graphs that are \emph{labelled} (\ie, the signature features several binary relations), or \emph{unlabelled} (\ie, 
only one binary relation).
\begin{result}[Sections~\ref{sec:labelled} and~\ref{sec:unlabelled}]
	The results in Table~\ref{tab:tab}, described in terms of the graph classes outlined above, hold.
\end{result}
In summary, for the classes that we consider, our results mostly show that the
general intractability of combined PQE carries over to the approximate
PQE problem.
The important exception is Proposition~\ref{prop:phoml-owp-dagi}:
the PQE problem for one-way path
queries on \emph{directed acyclic graphs} (DAGs) admits a combined FPRAS.
We discuss more in detail below how this result is proved and some of its
consequences. Another case is left open: in the unlabelled setting, we do not settle the approximability of combined PQE for one-way path queries (or equivalently downward tree queries) on
arbitrary graphs. For all other cases, either exact combined PQE was already
shown to be tractable in the exact setting~\cite{AMS17}, or we strengthen the
\SPhardness of exact PQE from~\cite{AMS17} by showing that combined FPRASes
conditionally do not exist.
We stress that our results always concern \emph{multiplicative approximations}: 
as non-probabilistic graph homomorphism is tractable for
the classes that we consider, we can always obtain additive approximations
for PQE simply by Monte Carlo sampling.
Further note that our intractability results are always
shown in \emph{combined complexity}---in data complexity, for the queries that
we consider, PQE is always multiplicatively approximable
via the Karp-Luby
algorithm~\cite{DBLP:series/synthesis/2011Suciu}.

As an important consequence, our techniques yield connections between
approximate PQE and \textit{intensional} approaches to the PQE problem.
Recall that the intensional approach was introduced by Jha and Suciu~\cite{JS13} in the
setting of exact evaluation, and when measuring data complexity. They show that
many tractable queries for PQE also admit tractable provenance
representations. More precisely, for these queries~$Q$, there is a polynomial-time
algorithm that takes as input any database instance and computes
a representation of the Boolean provenance of~$Q$ in a form which admits tractable model counting (\eg, OBDD, d-DNNF, etc.). This 
intensional approach
contrasts with \textit{extensional} approaches (like~\cite{DBLP:journals/jacm/DalviS12}) which exploit the structure of the query directly: comparing the power of intensional and extensional approaches is still an open problem~\cite{M20}.

In line with this intensional approach, we complement our conditional hardness results on approximate PQE with \emph{unconditional} lower bounds on the \textit{combined} size of tractable representations of query provenance.
Namely, we show a moderately exponential lower bound on DNNF provenance representations
which applies to all our non-approximable query-instance class pairs:
\begin{result}[Section~\ref{sec:dnnf-bounds}, informal]
	Let $\langle \Gmc, \Hmc \rangle$ be a conditionally non-approximable query-instance class pair studied in this paper. For any $\epsilon > 0$, there is an infinite family $G_1, G_2, \ldots$ of $\Gmc$ queries and an infinite family $H_1, H_2, \ldots$ of 
        $\Hmc$ instances such that, for every $i
	> 0$, every DNNF circuit representing the provenance $\prov^{G_i}_{H_i}$ has
        size at least $2^{\Omega \left( \left(\gsize{G_i} + \gsize{H_i} \right)^{1-\epsilon} \right)}$.
\end{result}
The class of DNNF circuits is arguably the most succinct circuit class in knowledge compilation that still has desirable properties~\cite{D01,DM02}.
Such circuits subsume in particular the class of \emph{structured DNNFs}, for which tractable approximation algorithms were recently proposed~\cite{DBLP:conf/stoc/ArenasCJR21a}.
Thus, these bounds help to better understand the limitations of intensional 
approaches.

Moreover, since we also show \textit{strongly} exponential lower bounds for the treewidth-1 query class \owp in particular, our results give an interesting example of a CQ class for which (non-probabilistic) query evaluation is in linear-time combined complexity~\cite{Y81}, but the size of every DNNF representation of query provenance is exponential:

\begin{result}[Proposition~\ref{prop:owp-all-lb}]
  \label{res:provenance}
	There exists an infinite family $G_1, G_2, \ldots$ of treewidth-1 CQs, and an infinite family $H_1, H_2, \ldots$ of instances on a fixed binary signature such that, for every $i > 0$, every DNNF circuit representing the provenance $\prov^{G_i}_{H_i}$ has size at least $2^{\Omega \left(\gsize{G_i} + \gsize{H_i} \right)}$.
\end{result}
Note that this result stands in contrast to the data complexity perspective, where the provenance of \textit{every} fixed CQ---no matter its treewidth---admits a polynomially-sized DNNF circuit representation, more precisely as a provenance formula in disjunctive normal form (DNF).

\paragraph*{Consequences.}
Our results and techniques have several interesting consequences of potential independent interest. First, they imply that we cannot relax the hypotheses of the result of van Bremen and Meel mentioned earlier~\cite{vBM23}. They show the following result on combined FPRASes for PQE in the more general context of probabilistic databases:

\begin{restatable}[Theorem 1 
of~\cite{vBM23}]{thm}{combinedsjfcq}\label{thm:combinedsjfcq}
	Let $Q$ be a self-join-free conjunctive query of bounded hypertree width, 
	and $H$ a tuple-independent database instance. %
        Then there exists a combined FPRAS for computing the probability of $Q$ 
        on $H$, \ie, an FPRAS whose runtime is $\poly(|Q|, \gsize{H}, 
        \epsilon^{-1})$, where $\epsilon$ is the multiplicative error.%
\end{restatable}

It was left open in~\cite{vBM23} whether intractability held without these assumptions on the query.
Hardness is immediate if we do not bound the width of queries and allow
arbitrary self-join-free CQs, as combined query evaluation is then \NPhard already in the non-probabilistic setting. However, it is less clear whether the self-join-freeness condition can be lifted. Our results give a negative answer, already in 
a severely restricted setting:

\begin{result}[Corollaries~\ref{cor:tw1tw1} and~\ref{cor:optimality}]
        Assuming $\RP \neq \NP$, neither the bounded hypertree width nor self-join-free condition in Theorem~\ref{thm:combinedsjfcq} can be relaxed:
        even on a fixed signature consisting of a single binary relation, there is no FPRAS to approximate the probability of an input treewidth-1 CQ on an input treewidth-1 TID instance.
\end{result}

A second consequence implied by our techniques concerns the \textit{two-terminal
network reliability problem} (ST-CON) on directed acyclic graphs (DAGs). Roughly
speaking, given a directed graph $G = (V, E)$ with independent edge reliability
probabilities $\pi: E \rightarrow [0, 1]$, and two distinguished vertices $s, t
\in V$, the ST-CON problem asks for the probability that there is a path from $s$ to $t$.
The problem is known to be \SPhard even on DAGs~\cite[Table~2]{PB83}.
The existence of an FPRAS for ST-CON is a long-standing open question~\cite{kannan1994markov}, %
and the case of DAGs was explicitly left open by Zenklusen and Laumanns~\cite{ZL11}. Our results allow us to answer in the affirmative:
\begin{result}[Theorem~\ref{thm:twoterminal-dag-fpras}]
	There exists an FPRAS for the ST-CON problem over DAGs.
\end{result}
This result and our approximability results follow from
the observation that path queries on directed acyclic graphs admit a compact representation of their Boolean provenance as \textit{non-deterministic ordered binary decision diagrams} (nOBDDs). We are then able to use a recent result by Arenas et al.~\cite[Corollary~4.5]{ACJR21} giving an FPRAS for counting the satisfying assignments of an nOBDD, adapted to the weighted setting. 

We last explore a third consequence of our work by studying the PQE problem
for \emph{regular path queries} (RPQs), i.e., queries asking for the 
existence of a walk in the graph (of arbitrary length) that forms a word 
belonging to a 
regular expression. Unlike the other queries studied in this work, RPQs are
generally not expressible as conjunctive queries, and so PQE is not 
necessarily approximable even in the data complexity perspective.
Thus, specifically for RPQs, we study PQE in data complexity rather than
combined complexity. Of course, some RPQs 
are equivalent to UCQs, for instance those with regular expressions 
describing a finite language, and so are approximable.
We show that, for all other RPQs (so-called \textit{unbounded} RPQs), the PQE 
problem is at least as hard as 
ST-CON in data complexity. For some unbounded RPQs, we also show
(conditional) inapproximability. This result does not directly follow from our
combined complexity results, but uses very similar techniques.

\paragraph*{Paper Structure.} 
In Section~\ref{sec:prelims}, we review some of the technical background.
We then present our main results on approximability, divided into the labelled and unlabelled case, in Sections~\ref{sec:labelled} and~\ref{sec:unlabelled} respectively. Next, in Section~\ref{sec:dnnf-bounds}, we show lower bounds on DNNF provenance circuit sizes. In Section~\ref{sec:applications}, we show some consequences for previous work~\cite{vBM23}, as well as for the two-terminal network reliability problem.
We show consequences for the data complexity of PQE for RPQs in
Section~\ref{sec:rpqs}.
We conclude in Section~\ref{sec:conclusions}.%

\section{Preliminaries}\label{sec:prelims}
We provide some technical background below, much of which closely follows that in~\cite{ACMS20} and~\cite{AMS17}.

\paragraph*{Graphs and Graph Homomorphisms.}
Let $\sigma$ be a non-empty finite set of labels called the \emph{signature}. When $|\sigma| = 1$, we say that we are in the \textit{unlabelled setting}; otherwise, we say we are in the \textit{labelled setting}. In this paper, we study only \textit{directed} graphs with edge labels from $\sigma$. A graph $G$ over $\sigma$ is a tuple $(V, E, \lambda)$ with finite non-empty 
vertex set~$V$,
edge set 
$E \subseteq V^2$, 
and $\lambda\colon E \to \sigma$ 
a labelling function mapping each edge to a single label (we may omit $\lambda$
in the unlabelled setting). The \emph{size} $\gsize{G}$ of $G$ is its number of edges.
We write $\ledge{x}{R}{y}$ 
for
an edge $e = (x,y) \in E$
with label~$\lambda(e) = R$,
and $\edge{x}{y}$ for $(x,y) \in E$ (no matter the edge label):
we say that $x$ is the \emph{source} of~$e$ and that $y$ is the \emph{target}
of~$e$. We sometimes use
a simple regular-expression-like syntax (omitting the vertex names where irrelevant) to represent
path graphs: for example, we write $\edge{}{}\edge{}{}$ to represent an unlabelled path of length two, and the notation $\edge{}{}^k$ to denote an unlabelled path of length $k$. All of this syntax extends to labelled graphs in the obvious way.
A graph $H = (V', E', \lambda')$ is a \textit{subgraph} of $G$,
written $H \subseteq G$, if $V = V'$, $E' \subseteq E$, and $\lambda'$ is the
restriction of~$\lambda$ to~$E'$.

A \textit{graph homomorphism} $h$ from a graph $G = (V_G, E_G, \lambda_G)$ to a graph $H = (V_H, E_H, \lambda_H)$ is a function $h: V_G \rightarrow V_H$ such that, for all $(u,v) \in E_G$, we have $(h(u),h(v)) \in E_H$ and $\lambda_H((h(u),h(v))) = \lambda_G((u,v))$. We write $G \homo H$ to say that such a homomorphism exists, and sometimes refer to a homomorphism from $G$ to $H$ as a \textit{match} of $G$ in $H$.

\paragraph*{Probabilistic Graphs and Probabilistic Graph Homomorphisms.}
A \textit{probabilistic graph} is a pair $(H, \pi)$, where $H$ is a graph with edge labels from $\sigma$, and $\pi: E \to [0, 1]$ is a probability labelling on the edges. Note that edges $e$ in $H$ are annotated both by their probability value $\pi(e)$ and their $\sigma$-label $\lambda(e)$. Intuitively, $\pi$ gives us a succinct specification of a probability distribution over the $2^{\gsize{H}}$ possible subgraphs of $H$, by independently including each edge~$e$ with 
probability $\pi(e)$.
Formally, the distribution induced by $\pi$ on the subgraphs $H' \subseteq H$ is defined by $\Pr_\pi(H') = \prod_{e \in E'} \pi(e) \prod_{e \in E \setminus E'} (1 - \pi(e))$.

In this paper, we study the \textit{probabilistic graph homomorphism} problem \phom: given a graph $G$ called the \emph{query graph} and a probabilistic graph $(H, \pi)$ called the \emph{instance graph}, with $G$ and $H$ carrying labels from the same signature $\sigma$, we must compute the probability $\Pr_\pi(G \homo H)$ that a subgraph of $H$, sampled 
according to the distribution induced by $\pi$, admits a homomorphism from $G$. That is, we must compute
$\Pr_\pi(G \homo H) := \sum_{\substack{H' \subseteq H \text{~s.t.~} G \homo H'}} \Pr_\pi(H')$.

We study \phom in \emph{combined complexity}, i.e., when the query graph $G$, instance graph $(H, \pi)$, and signature $\sigma$ are all given as input. Further, we study \phom when we restrict $G$ and $H$ to be taken from specific \emph{graph classes}, i.e., infinite families of (non-probabilistic) graphs, denoted respectively $\Gmc$ and $\Hmc$. (Note that $\Hmc$ does not restrict the probability labelling $\pi$.)
To distinguish the \textit{labelled} and \textit{unlabelled} setting, we denote by $\phoml(\Gmc, \Hmc)$ the problem of computing $\Pr_\pi(G \homo H)$ for $G \in \Gmc$ and $(H, \pi)$ with $H \in \Hmc$ when no restriction is placed on the input signature $\sigma$ for graphs in $\Gmc$ and $\Hmc$. On the other hand, we write $\phomul(\Gmc, \Hmc)$ when $\Gmc$ and $\Hmc$ are restricted to be classes of unlabelled graphs, i.e., $|\sigma| = 1$.
We focus on \emph{approximation algorithms}: fixing classes $\Gmc$ and~$\Hmc$, a \emph{fully polynomial-time randomized approximation scheme} (FPRAS) for $\phoml(\Gmc, \Hmc)$ (in the labelled setting) or $\phomul(\Gmc, \Hmc)$ (in the unlabelled setting) is a randomized algorithm that runs in time $\poly(\gsize{G}, \gsize{H}, \epsilon^{-1})$ on inputs
$G \in \Gmc$, $(H, \pi)$ for $H \in \Hmc$, and $\epsilon > 0$.
The algorithm must return, with probability at least $3/4$, a \emph{multiplicative approximation} of the probability $\Pr_\pi(G \homo H)$, i.e., a value between $(1-\epsilon) \Pr_\pi(G \homo H)$ and $(1+\epsilon) \Pr_\pi(G \homo H)$.

\paragraph*{Graph Classes.}
We study \phom on the following graph classes, which are defined on a graph $G$ with edge labels from some signature $\sigma$, and thus can either be labelled or unlabelled depending on $\sigma$:
\begin{itemize}
	\item $G$ is a \textit{one-way path} (\owp) if it is of the form $\ledge{a_1}{R_1}{} \ledge{\dots}{R_{m-1}}{a_m}$ for some $m$, with all $a_1, \dots, a_m$ being pairwise distinct, and with $R_i \in \sigma$ for $1 \leq i < m$.
	\item $G$ is a \textit{two-way path} (\twp) if it is of the form $a_1\ -\ \dots\ -\ a_m$ for some $m$, with pairwise distinct $a_1, \dots, a_m$, and each $-$ being $\ledge{}{R_i}{}$ or $\redge{}{R_i}{}$ (but not both) for some label $R_i \in \sigma$.
	\item $G$ is a \textit{downward tree} (\dwt) if it is a rooted unranked tree (each node can have an arbitrary number of children), with all edges pointing from parent to child in the tree.
	\item $G$ is a \textit{polytree} (\pt) if its underlying undirected graph is a rooted unranked tree, without restrictions on the edge directions.
	\item $G$ is a \textit{DAG} (\dagi) if it is a
          (directed) acyclic graph.
\end{itemize}
These refine the classes of connected queries considered in~\cite{AMS17}, by
adding the \dagi class. We denote by \all the class of all graphs.
Note that both \twp and \dwt generalize \owp and are
incomparable; \pt generalizes both \twp and \dwt; \dagi generalizes \pt; \all
generalizes \dagi (see Figure~2 of~\cite{AMS17}).

Note that our notion of labelled graphs, and the classes of graphs defined
above, are different from the notion of database instances, where we would
typically allow two vertices to be connected by multiple edges with different 
labels. Formally, for a non-empty finite signature $\sigma$, an
\emph{arity-two database} over~$\sigma$ consists of an \emph{active domain} $V$ and a
set $E$ of labelled edges of the form $(u, a, v)$ with $u, v \in V$ and $a \in
\sigma$. The notion of a probabilistic arity-two database is defined in the
expected way by giving a probability to every edge of~$E$ and assuming
independence across all edges (in particular all edges having the same endpoints are also independent).
A graph over~$\sigma$ can be seen as an arity-two database where for
every pair $(u,v) \in V \times V$ there exists at most one edge in~$E$ of the
form $(u, a, v)$ for some $a \in \sigma$. Note that in the unlabelled setting
with $|\sigma|=1$ there is no difference between both notions. For simplicity,
all our results will be phrased in terms of graphs and not of arity-two
databases; in particular, all our lower bounds will apply in the restricted
setting of graphs that we study.
However, all our upper bound results will in fact also hold when allowing
arity-two databases as input---we mention this explicitly when stating these results.

\paragraph*{Boolean Provenance.}
We use the notion of \textit{Boolean provenance}, or simply
\textit{provenance}~\cite{IL84,ABS15,S19}.
In the context of databases, provenance intuitively represents which subsets of
the instance satisfy the query: it is used in the intensional approach to
probabilistic query
evaluation~\cite{JS13}.
In this paper, we use provenance to show both upper and lower bounds.

Formally,
let $G = (V_G, E_G, \lambda_G)$ and $H = (V_H, E_H, \lambda_H)$ be graphs. 
Seeing $E_H$ as a set of Boolean variables, a \emph{valuation} $\nu$ of~$E_H$ is a
function $\nu\colon E_H \to \{0, 1\}$ that maps each edge of~$H$ to~$0$ or~$1$.
Such a valuation~$\nu$ defines a subgraph $H_\nu$ of~$H$ where we only keep
the edges mapped to~$1$, formally $H_\nu = (V_H, \{e \in E_H\ |\ \nu(e) = 1\},
\lambda_H)$.
The \emph{provenance} of~$G$ on~$H$ is then the Boolean function
$\prov^G_H$ having as variables the edges~$E_H$ of~$H$ and mapping
every valuation $\nu$ of~$E_H$ 
to~$1$ (true) or $0$ (false) depending on whether $G \homo H_\nu$ or not.
Generalizing this definition,
for any integer $n$, for any choice of $a_1, \ldots, a_n \in V_G$ and $b_1, \ldots, b_n \in V_H$, we write
$\prov^G_H[a_1 \colonequals b_1, \dots, a_n \colonequals b_n]$ 
to denote the Boolean function that maps valuations~$\nu$ of~$E_H$ to~$1$
or~$0$ depending on whether or not there is a homomorphism $h\colon G \to H_\nu$ which additionally satisfies $h(a_i) = b_i$ for all $1 \leq i \leq n$.

For our lower bounds, we will often seek to represent Boolean formulas as the
provenance of queries on graphs:

\begin{defi}
  \label{def:provrepresents}
	Given two graphs $G$ and~$H$, and a Boolean formula $\phi$ whose variables $\{e_1, \ldots, e_n\} \subseteq E_H$ are edges of $H$, we say that $\prov_H^G$ represents $\phi$ on $(e_1,...,e_n)$ if for every valuation $\nu: E_H \rightarrow \{0,1\}$ that maps edges not in $\{e_1,...,e_n\}$ to $1$, we have $\nu\models\phi$ if and only if $\prov_H^G(\nu)=1$.
\end{defi}

\paragraph*{Circuits and Knowledge Compilation.}
We consider representations of Boolean functions in terms of \textit{non-deterministic (ordered) binary decision diagrams}, as well as \textit{decomposable circuits}, which we define below.

A \textit{non-deterministic binary decision diagram} (nBDD) on a set of variables $V = \{v_1, \dots, v_n\}$ is a rooted DAG $D$ whose nodes carry a label in $V \sqcup \{0, 1, \lor\}$ (using $\sqcup$ to denote disjoint union), and whose edges can carry an optional label in $\{0, 1\}$ subject to the following requirements:
\begin{enumerate}
  \item there are exactly two leaves (called \textit{sinks}), one labelled by $1$ (the \textit{$1$-sink}), and the other by $0$ (the \emph{$0$-sink});
  \item internal nodes are labelled either by $\lor$ (called an \emph{$\lor$-node}) or by a variable of $V$ (called a \emph{decision node}); and
	\item each decision node has exactly two outgoing edges, labelled $0$ and $1$; the outgoing edges of $\lor$-nodes are unlabelled.
\end{enumerate}
The size $\gsize{D}$ of $D$ is its number of edges. Let $\nu$ be a valuation of $V$, and let $\pi$ be a path in $D$ going from the root to one of the sinks. We say that $\pi$ is \textit{compatible} with $\nu$ if for every decision node $n$ of the path, letting $v \in V$  be the variable labelling~$n$, then $\pi$ passes through the outgoing edge of $n$ labelled with $\nu(v)$. In particular, no constraints are imposed at $\lor$-nodes; thus, we may have that multiple paths are compatible with a single valuation.
The nBDD $D$ \emph{represents} a Boolean function, also written $D$ by abuse of notation, which is defined as follows:
for each valuation~$\nu$ of~$V$, we set $D(\nu) \colonequals 1$ if %
there exists a path $\pi$ from the root to the $1$-sink of $D$ that is compatible with $\nu$, and set $D(\nu) \colonequals 0$ otherwise.
Given an nBDD $D$ over variables $V$, we denote by $\mods(D)$ the set of satisfying valuations $\nu$ of $D$ such that $D(\nu) = 1$, and by $\mc(D)$ the number $|\mods(D)|$ of such valuations. Further, given a rational probability function $w: V \rightarrow [0, 1]$ on the variables of $V$, define $\wmc(D, w)$ to be the probability that a random valuation $\nu$ satisfies $D$, that is, $\wmc(D, w) = \sum_{\nu \in \mods(D)} \prod_{x \in V \text{~s.t.~} \nu(x) = 1} w(x) \prod_{x \in V \text{~s.t.~} \nu(x) = 0} \left( 1 - w(x) \right)$.

In this paper, we primarily focus on a subclass of nBDDs called
\textit{non-deterministic ordered binary decision diagrams} (nOBDDs). An nOBDD $D$ is an nBDD 
for which there exists a strict total order $\prec$ on the variables~$V$ such that, for any two decision nodes $n \neq n'$ such that there is a path from~$n$ to~$n'$, then, letting $v$ and $v'$ be the variables that respectively label~$n$ and~$n'$, we have $v \prec v'$.
This implies that, along any path going from the root to a sink, the sequence of variables will be ordered according to~$V$, with each variable occurring at most once.
We use nOBDDs because they 
admit tractable approximate counting of their satisfying assignments, as we discuss later.

We also show lower bounds on a class of \emph{circuits}, called \emph{decomposable negation normal form} (DNNF) circuits. A \emph{circuit} on a set of variables~$V$ is a directed acyclic graph $C = (G, W)$, where $G$ is a set of \emph{gates}, where~$W \subseteq G \times G$ is a set of edges called \emph{wires}, and where we distinguish an \emph{output gate} $g_0 \in G$. The \emph{inputs} of a gate~$g \in G$ are the gates $g'$ such that there is a wire $(g', g)$ in~$W$. The gates can be labelled with variables of~$V$ (called a \emph{variable gate}), or with the Boolean operators $\lor$, $\land$, and $\neg$. We require that gates labelled with variables have no inputs, and that gates labelled with~$\neg$ have exactly one input. A circuit $C$ defines a Boolean function on~$V$, also written $C$ by abuse of notation. Formally, given a valuation $\nu$ of~$V$, we define inductively the \emph{evaluation} $\nu'$ of the gates of $C$ by setting $\nu'(g) \colonequals \nu(v)$ for a variable-gate $g$ labelled with variable $v$, and setting $\nu'(g)$ for other gates to be the result of applying the Boolean operators of~$g$ to $\nu'(g_1), \ldots, \nu'(g_n)$ for the inputs $g_1, \ldots, g_n$ of~$g$. We then define $C(\nu)$ to be $\nu'(g_0)$ where $g_0$ is the output gate of~$C$.

The circuit is in \emph{negation normal form} if negations are only applied to variables, i.e., for every $\neg$-gate, its input is a variable gate. The circuit is \emph{decomposable} if the $\land$-gates always apply to inputs that depend on disjoint variables: formally, there is no $\land$-gate $g$ with two distinct inputs $g_1$ and $g_2$, such that some variable $v$ labels two variable gates $g_1'$ and $g_2'$ with $g_1'$ having a directed path to~$g_1$ and $g_2'$ having a directed path to~$g_2$. A \emph{DNNF} is a circuit which is both decomposable and in negation normal form.
Note that we can translate nOBDDs in linear time to 
DNNFs, more specifically to \emph{structured DNNFs}~\cite[Proposition~3.8]{ACMS20}.

\paragraph*{Approximate Weighted Counting for nOBDDs.}%
Recently, Arenas et al.~\cite{ACJR21} showed the following result on approximate counting of satisfying assignments of an nOBDD.

\begin{thm}[Corollary~4.5 of \cite{ACJR21}]\label{thm:nobdd-unweighted-count}
	Let $D$ be an nOBDD. Then there exists an FPRAS for computing $\mc(D)$.
\end{thm}

For our upper bounds, we need a slight strengthening of this result to apply to \emph{weighted model counting} ($\wmc$) in order to handle probabilities. This can be achieved by translating the approach used in \cite[Section~5.1]{vBM23} to the nOBDD setting. We thus show (see Appendix~\ref{apx:nobdd-weighted-count}):

\begin{restatable}{thm}{nobddweightedcount}\label{thm:nobdd-weighted-count}
  Let $D$ be an nOBDD on variables $V$, and let $w: V \rightarrow [0, 1]$ be a rational probability function defined on $V$. Then there exists an FPRAS for computing $\wmc(D, w)$, running in time polynomial in $\gsize{D}$ and~$w$.
\end{restatable}

\begin{table}[t]
        \renewcommand{\tabcolsep}{2pt}
		\newcommand\headercell[1]{\smash[b]{\begin{tabular}[t]{@{}c@{}} #1 \end{tabular}}}
\captionsetup[subtable]{justification=centering}
		\begin{subtable}{0.5\textwidth}
                  \centering
			\begin{tabular}{@{~}c | *{6}{c} @{\hspace{2pt}} }
				\headercell{$\Gmc \downarrow$} & \multicolumn{6}{c@{}}{$\Hmc \rightarrow$ }\\
                                                               & \owp   &  \twp           & \dwt    & \,\,\pt\null     & \dagi   & \,\,\,\all\,   \\ 
				\hline
				\owp  &        &                 &         &  \cellcolor{lightgray} & \cellcolor{lightgray}\textcolor{white}{\ref{prop:phoml-owp-dagi}} &  \cellcolor{darkgray}\textcolor{white}{\ref{prop:phoml-owp-all}} \\
				\twp  &        &     &  \cellcolor{darkgray}\textcolor{white}{\ref{prop:phoml-twp-dwt}} &  \cellcolor{darkgray} & \cellcolor{darkgray} &  \cellcolor{darkgray} \\
				\dwt   &        &     &  \cellcolor{darkgray}\textcolor{white}{\ref{prop:phoml-dwt-dwt}} &  \cellcolor{darkgray} & \cellcolor{darkgray} &  \cellcolor{darkgray} \\
				\pt    &        &     &  \cellcolor{darkgray} &  \cellcolor{darkgray} & \cellcolor{darkgray} &  \cellcolor{darkgray} \\
			\end{tabular}
			\caption{Complexity of $\phoml(\Gmc, \Hmc)$.}\label{tab:labelled}
		\end{subtable}\hfill
		\begin{subtable}{0.5\textwidth}
                  \centering
                  \begin{tabular}{@{~} c | *{6}{c} @{\hspace{2pt}} }
				\headercell{$\Gmc \downarrow$} & \multicolumn{6}{c@{}}{$\Hmc
					\rightarrow$ }\\
                                                               & \owp   &  \twp           & \dwt    & \,\,\pt\null     & \dagi   & \,\,\,\all\,   \\ 
				\hline
				\owp  &        &                 &         &   & \cellcolor{lightgray}\textcolor{white}{\ref{prop:phomul-owp-dagi}} & ? \\
				\twp  &        &     &   & \cellcolor{darkgray}\textcolor{white}{\ref{prop:phomul-twp-pt}}  & \cellcolor{darkgray} &  \cellcolor{darkgray} \\
				\dwt   &        &     &   &   & \cellcolor{lightgray}\textcolor{white}{\ref{prop:phomul-dwt-dagi}} &  ? \\
				\pt    &        &     &   &  \cellcolor{darkgray} & \cellcolor{darkgray} &  \cellcolor{darkgray} \\
			\end{tabular}
			\caption{Complexity of $\phomul(\Gmc, \Hmc)$.}\label{tab:unlabelled}
		\end{subtable}
		\caption{Results on approximation proved in this paper. Key: white (\crule[white]{0.5\baselineskip}{0.5\baselineskip}) means that the problem is in $\PTIME$; light grey (\crule[lightgray]{0.5\baselineskip}{0.5\baselineskip}) means that it is \SPhard but admits an FPRAS; dark grey (\crule[darkgray]{0.5\baselineskip}{0.5\baselineskip}) means \SPhardness and non-existence of an FPRAS, assuming $\RP \neq \NP$. All cells without a reference to a corresponding proposition are either implied by one of the other results in this paper, or pertain to exact complexity and were already settled in~\cite{AMS17}.}\label{tab:tab}
\end{table}

\section{Results in the Labelled Setting}\label{sec:labelled}

We now move on to the presentation of our results. We start with the
\textit{labelled} setting of probabilistic graph
homomorphism, in which no restriction is placed on the signature $\sigma$ of the query and instance graphs.
Our results are summarized in Table~\ref{tab:labelled}. Note that the \textit{exact} tractability and hardness results shown in \cite{AMS17} (and repeated in Table~\ref{tab:labelled}) actually pertain to a slightly more restricted setting in which the signature $\sigma$ is fixed, different to the setting studied here in which $\sigma$ forms part of the problem input. Fortunately, it is easy to check that all results still hold: for the \SPhardness results, this is immediate, and for the $\PTIME$ tractability results, a careful inspection of the relevant claims (\cite[Proposition~4.10]{AMS17} and \cite[Proposition~4.11]{AMS17}) shows that tractability safely carries over to this more general setting. We can therefore focus on proving (in)tractability of approximation results in the remainder of this paper.

\paragraph*{$\Bowp$ on $\Bdagi$.}
We start by showing the tractability of approximation for $\phoml(\owp, \dagi)$, which also implies tractability of approximation for $\phoml(\owp, \pt)$, since $\pt \subseteq \dagi$. 
\begin{prop}\label{prop:phoml-owp-dagi}
	$\phoml(\owp, \dagi)$ is \SPhard already in data complexity, but it
        admits an FPRAS. This holds even if the input instance is a
        probabilistic arity-two
        database.
\end{prop}
For \SPhardness, the result already holds in the unlabelled setting, so it
will be shown in Section~\ref{sec:unlabelled} (see
Proposition~\ref{prop:phomul-owp-dagi}). Hence, we focus on the upper bound.
We rely on the notion of a \emph{topological ordering} of the edges of a
directed acyclic graph $H = (V, E)$: it is simply a strict total order $(E,
\prec)$ with the property that for every consecutive pair of edges $e_1 =
(a_1,a_2)$ and $e_2 = (a_2,a_3)$, we have that $e_1 \prec e_2$. Let us fix such
an ordering.
\begin{proof}[Proof of Proposition~\ref{prop:phoml-owp-dagi}]
  We will show that every \owp query on a \dagi instance (possibly an
  arity-two database, i.e., with multiple edges having the same
  endpoints) admits an nOBDD
	representation of its provenance, which we can compute in combined
        polynomial time. We can then apply
        Theorem~\ref{thm:nobdd-weighted-count}, from which the result follows.
        Let $G = \ledge{a_1}{R_1}{} \ledge{\dots}{R_{m}}{a_{m+1}}$ be the input
        $\owp$ query, and $H = (V, E)$ the instance arity-two database.
	We make the following claim:
	\begin{clm}
          For every $v \in V$, we can compute in time $O(\gsize{G} \times \gsize{H})$ an nOBDD
          representing $\provl^G_H[a_1 \colonequals v]$ which is ordered by the
          topological ordering $\prec$ fixed above.
	\end{clm}
	\begin{proof}
          We build an nBDD $D$ consisting of the two sinks and of the following nodes:
          \begin{itemize}
            \item $|V| \times \gsize{G}$ $\lor$-nodes written $n_{u,i}$ for $u \in V$ and $1 \leq
          i \leq m$; and
        \item  $|E| \times \gsize{G}$ decision nodes written $d_{e,i}$ for
          $e \in E$ and $1 \leq
          i \leq m$ which test the edge~$e$.
          \end{itemize}
          Each $\lor$-node $n_{u,i}$ for $u \in V$ and $1 \leq i \leq m$ has
          outgoing edges to each $d_{e,i}$ for every edge~$e$ emanating from~$u$
          which is labelled $R_i$. For each decision node $d_{e,i}$, letting $w$ be the target
          of edge~$e$, then $d_{e,i}$ has an outgoing $0$-edge to the
          $0$-sink and an outgoing $1$-edge to either $n_{w,i+1}$ if $i<m$ or to
          the $1$-sink if $i=m$. The root of the nBDD is the node $n_{v,1}$.

          This construction clearly respects the time bound.
          To check correctness of the resulting nBDD, it is immediate to observe
          that, for any
          path from the root to a sink, the sequence of decision nodes traversed
          is of the form $d_{e_1, 1}, \ldots, d_{e_k, k}$ where the $e_1,
          \ldots, e_k$ form a path of consecutive edges starting at~$v$ and
          successively labelled $R_1, \ldots, R_k$. This implies that the nBDD is
          in fact an nOBDD ordered by~$\prec$. Further, such a path reaches the
          $1$-sink iff $k=m$ and all decisions are positive, which implies that
          whenever the nOBDD accepts a subgraph $H'$ of~$H$ then indeed
          there is a match of~$G$ in~$H'$
          that maps $a_1$ to~$v$.
          For the converse
          direction, we observe that, for any subgraph $H'$ of~$H$ 
          such that there is a match of~$G$ in~$H'$
          mapping $a_1$ to~$v$, then, letting $e_1, \ldots, e_m$ be
          the images of the successive edges of~$G$ in~$H'$, there is a path
          from the root of~$D$ to the $1$-sink which tests these edges in order.
          This establishes correctness and concludes the proof of the claim.
	\end{proof}
	Now observe that $\provl^G_H = \provl^G_H[a_1 \colonequals v_1]\ \lor\
	\cdots\ \lor\ \provl^G_H[a_1 \colonequals v_n]$, where $v_1, \dots, v_n$ are precisely the vertices of $H$. Thus, it suffices to simply take the disjunction of each nOBDD obtained using the process above across every vertex in $H$, which yields in linear time the desired nOBDD.
	From here we can apply Theorem~\ref{thm:nobdd-weighted-count}, concluding the proof. %
\end{proof}
\paragraph*{$\Bowp$ on Arbitrary Graphs.}
We show, however, that tractability of approximation does \textit{not} continue to hold when relaxing the instance class from $\dagi$ to arbitrary graphs. This also implies that more expressive classes of query graphs---such as \twp, \dwt, and \pt also cannot be tractable to approximate on instances in the class \all.

\begin{prop}\label{prop:phoml-owp-all}%
	$\phoml(\owp, \all)$ does not admit an FPRAS unless $\RP=\NP$.
        This holds even for a fixed signature consisting of two labels.
\end{prop}
\begin{proof}
	Our result hinges on the following claim:
	\begin{clm}
	\label{clm:2cnf2prov}
	Let $d > 1$ be a constant and let $\sigma$ be a fixed signature with at least two labels.
	Given a monotone 2-CNF formula $\phi$ on $n$ variables where each variable occurs in at most $d$ clauses, we can build in time $O(|\phi|)$ a \owp $G_\phi$ and graph $H_\phi$ in the class \all and over signature~$\sigma$ containing edges $(e_1, \ldots, e_n)$ such that $\prov^{G_\phi}_{H_\phi}$ represents $\phi$
	on~$(e_1, \ldots, e_n)$.
	\end{clm}
	\begin{proof}
	Let $\phi = \bigwedge_{1 \leq i \leq m} (X_{f_1(i)} \lor X_{f_2(i)})$ be
          the input CNF instance over the variables $\{X_1, \dots, X_n\}$,
          where $m > 0$ is the number of clauses. As we are in the labelled setting, 
          let $U$ and $R$ be two distinct labels from the signature. Define the \owp query graph $G_\phi$ to be $\ledge{}{U}{} \ledge{}{U}{} \left( \ledge{}{R}{}^{d+2} \ledge{}{U}{} \right)^m \ledge{}{U}{}$.	
	The instance graph $H_\phi$ in the class \all is defined in the following way:
	\begin{itemize}
		\item For all $1 \leq i \leq n$, add an edge $\ledge{a_i}{R}{b_i}$.
		\item Add an edge $\ledge{c_0}{U}{d_0}$ and for each clause $1 \leq j \leq m$, an edge $\ledge{c_j}{U}{d_j}$.
                \item Add two edges $\ledge{c_0'}{U}{c_0}$ and
                  $\ledge{d_m}{U}{d_m'}$
		\item For each clause $1 \leq j \leq m$ and variable $X_i$ occurring in that clause, let $p$ be the number of this occurrence of $X_i$ in the formula (\ie, the occurrence of $X_i$ in the $j$-th clause is the
		$p$-th occurrence of $X_i$), with $1 \leq p \leq d$ by assumption on $\phi$.
		Then add a path of length $p$ of $R$-edges from $d_{j-1}$ to $a_i$ and a
		path of length $(d+1)-p$ of $R$-edges from $b_i$ to $c_j$.
	\end{itemize}
          The construction of $G_\phi$ and $H_\phi$ is in $O(|\phi|)$.
          Furthermore, notice the following~($\star$). For any $1 \leq i \leq n$, the
          edge $e = \ledge{a_i}{R}{b_i}$
	has at most $d$ incoming $R$-paths and $d$ outgoing $R$-paths; the
          outgoing paths have pairwise distinct \emph{length} (\ie, the number
          of edges until the next edge is a $U$-edge), and likewise for the incoming
	paths. What is more, each incoming $R$-path of length $p$ corresponds to an
	outgoing path of length $(d+1)-p$ and together they connect some $d_{j-1}$ to
	some $c_j$ via the edge~$e$, where the $j$-th clause contains variable $X_i$.

          Moreover, notice the following~($\star\star$): the only two pairs
          of contiguous $U$-edges are $\ledge{c_0'}{U}{} \ledge{c_0}{U}{d_0}$
          and $\ledge{c_m}{U}{} \ledge{d_m}{U}{d_m'}$.

	Now, define $(e_1, \dots, e_n)$ to be precisely the edges of the form
          $\ledge{a_i}{R}{b_i}$ for every $1 \leq i \leq n$. Intuitively, the presence
          or absence of each of these edges corresponds to the valuation of each variable in $\phi$.
	We claim that $\prov^{G_\phi}_{H_\phi}$ represents $\phi$
	on~$(e_1, \ldots, e_n)$. 
          Call a subgraph of~$H_\phi$ a \emph{possible world} if it contains all the edges not in
          $(e_1, \dots, e_n)$ (as these are fixed to~$1$).
          It will suffice to show that there is a
          bijection between the satisfying valuations of $\phi$ and the possible
          worlds of $H_\phi$ that admit a homomorphism from $G_\phi$.

          Indeed, consider the bijection defined in the obvious way: keep the edge $\ledge{a_i}{R}{b_i}$ iff $X_i$ is assigned to true in the valuation. First suppose that some valuation of $\{X_1, \dots, X_n\}$ satisfies $\phi$. Then,
          for each clause $1 \leq j \leq m$,
          there is a variable in the clause which evaluates to true.
	We build a match of $G_\phi$ on the corresponding possible world of
          $H_\phi$ as follows:
          \begin{itemize}
            \item map the leftmost $U$-edge to $\ledge{c_0'}{U}{c_0}$,
            \item map the rightmost $U$-edge to $\ledge{d_m}{U}{d_m'}$,
            \item for the other $U$-edges, map the
          $j$-th of these $U$-edges to $\ledge{c_j}{U}{d_j}$ for all $0 \leq j \leq m$,
        \item map the $R$-paths for each $1 \leq j \leq m$ by picking a variable $X_i$
	witnessing that the clause is satisfied and going via the path of length
	$1 + \left( p \right) + \left( \left( d+1 \right) - p \right) = d+2$ that uses the edge $\ledge{a_i}{R}{b_i}$, which is present by
	assumption. %
          \end{itemize}
	
	Conversely, assume that we have a match of $G_\phi$ on a possible world of $H_\phi$. We show that the corresponding valuation satisfies $\phi$.
          It is an easy consequence of ($\star\star$) that the first $U$-edge
          must be mapped to $\ledge{c_0'}{U}{c_0}$, so the second $U$-edge must
          be mapped to $\ledge{c_0}{U}{d_0}$. 
          Let us show by finite
          induction on $0 \leq j \leq m$ that the $(j+2)$-th $U$-edge must be
          mapped to $\ledge{c_j}{U}{d_j}$ and that if $j\geq 1$ the $j$-th
          clause is satisfied. The base case of $j=0$ is clear because the second
          $U$-edge is mapped to $\ledge{c_0}{U}{d_0}$.

        Let us take $j\geq 1$ and show the induction step. By induction
        hypothesis,
        the $(j+1)$-th $U$-edge is
        mapped to $\ledge{c_{j-1}}{U}{d_{j-1}}$. Now,
        the $R$-path that
        follows must be mapped to a path from $d_{j-1}$ to some $a_i$, and then take the
	edge $\ledge{a_i}{R}{b_i}$, whose presence witnesses that the corresponding
	variable $X_i$ is true. But importantly, in order for the path to have
	length precisely $d+2$ before reaching another $U$-edge, it must be the case that
	the length of the path before and after the edge $\ledge{a_i}{R}{b_i}$ sums up to
	$d+1$. As a result of ($\star$), this is only possible by taking a path that leads to
	$\ledge{c_{j}}{U}{d_{j}}$.
        Thus we know that variable $X_i$ occurs in the $j$-th
        clause so that this clause is satisfied, and we know that the $(j+2)$-th
        $U$-edge is mapped to $\ledge{c_{j}}{U}{d_{j}}$. This establishes the
        induction step. Thus, the inductive proof establishes that all clauses
        are satisfied. This establishes the converse direction of the
        correctness claim, and concludes the proof.
	\end{proof}
        We can now conclude the proof of Proposition~\ref{prop:phoml-owp-all}.
	By \cite[Theorem~2]{S10}, counting the independent sets of a graph of maximal degree $6$ admits an FPRAS only if $\RP = \NP$. 
	It is not hard to see that this problem is equivalent to counting satisfying assignments of a monotone 2-CNF formula in which a variable can appear in up to $6$ clauses (see, for example, \cite[Proposition~1.1]{LL15}). 
	Thus, given a monotone 2-CNF formula $\phi$ obeying this restriction, we can apply Claim~\ref{clm:2cnf2prov} above for the class of formulas in which $d = 6$ to obtain (deterministic) graphs $G_\phi$ and $H_\phi$ on a fixed signature $\sigma$ with two labels, and then build a probabilistic graph $H_\phi'$ identical to $H_\phi$, in which the edges $(e_1, \dots, e_n)$ are assigned probability $0.5$ and all other edges probability $1$. Now, any FPRAS for the probabilistic graph homomorphism problem on~$G_\phi$ and~$H_\phi'$ would give an approximation of $N/2^n$, where $N$ is the number of satisfying assignments of~$\phi$ and $n$ is the number of variables of~$\phi$. Simply multiplying the result of the FPRAS by $2^n$ would give us an FPRAS to approximate~$N$, and so would imply $\RP = \NP$, concluding the proof.
\end{proof}

\paragraph*{$\Bdwt$ on $\Bdwt$.}
Having classified the cases of one-way path queries (\owp) on all instance
classes considered, we turn to more expressive queries. The next two query classes to consider are two-way path queries ($\twp$) and downward trees
queries ($\dwt$). For these query classes, exact computation on $\twp$ instances is
tractable by~\cite{AMS17}, so the first case to classify is that of \dwt instances. Exact computation is intractable in this case
by~\cite{AMS17}, and we show here that, unfortunately, approximation is
intractable as well, so that the border for exact tractability coincides with
that for approximate tractability.
We first focus on $\dwt$ queries:

\begin{prop}\label{prop:phoml-dwt-dwt}
	$\phoml(\dwt, \dwt)$ does not admit an FPRAS unless $\RP=\NP$.
This holds even for a fixed signature consisting of two labels.
\end{prop}
\begin{proof}
        Our result hinges on the following, whose proof adapts~\cite[Proposition~2.4.3]{M18}:
        \begin{clm}
        	\label{clm:2cnf2prov-dwt}
	Let $\sigma$ be a fixed signature with at least two labels.
        	Given a monotone 2-CNF formula $\phi$ on $n$ variables, we can build in time 
                $O(|\phi| \log |\phi|)$
                \dwt graphs $G_\phi$ and $H_\phi$ over~$\phi$, with the latter containing edges $(e_1, \ldots, e_n)$ such that $\prov^{G_\phi}_{H_\phi}$ represents $\phi$
        	on~$(e_1, \ldots, e_n)$.
        \end{clm}

        \begin{proof}
		Let $\phi = \bigwedge_{1 \leq i \leq m} (X_{f_1(i)} \lor X_{f_2(i)})$
                be the input CNF instance over the variables $\{X_1, \dots,
                X_n\}$. We let $L = \lceil \log_2 m \rceil$ be the number of
                bits needed to write clause numbers in binary.
		As we are in the labelled setting, let $0$ and $1$ be two
                distinct labels from the signature.
		Construct the query graph $G_\phi$ as having a root $z$ and child nodes $c_1, \ldots, c_m$ corresponding to the clauses and having a child path of length~$L$ labelled by the clause number. Formally:
		\begin{itemize}
			\item For all $1 \leq i \leq m$, add an edge $\ledge{z}{0}{c_i}$.
            \item For each $1 \leq i \leq m$, letting $b_1 \cdots b_L$ be the clause number $i$ written in binary, add a path of $L$ edges 
                          $\ledge{c_{i}}{b_1}{d_{i,1}} \ledge{}{b_2}{\ldots}
                          \ledge{}{b_{L-1}}{} \ledge{d_{i,L-1}}{b_L}{d_{i,L}}$.
		\end{itemize}
        Now, construct the $\dwt$ instance $H_\phi$ as having a root $z'$ and child nodes $x_1, \ldots, x_n$ corresponding to the variables, each variable node $x_i$ having child paths of length $L$ labelled by the numbers of the clauses made true when setting $x$ to true. Formally:
	\begin{itemize}
			\item For all $1 \leq i \leq n$, add the edges $\ledge{z'}{0}{x_i}$.%
			\item For all $1 \leq i \leq n$ and $1 \leq j \leq m$ such that $X_i$ occurs in the $j$-th clause of $\phi$ (\ie, $i = f_1(j)$ or $i = f_2(j)$),
                          letting $b_1 \cdots b_L$ be the clause number $j$ written in binary, add a path of $L$ edges
$\ledge{x_i}{b_1}{y_{i,j,1}} \ledge{}{b_2}{\ldots} \ledge{}{b_{L-1}}{} \ledge{y_{i,j,L-1}}{b_L}{y_{i,j,L}}$.
	\end{itemize}
          It is clear that $G_\phi \in \dwt$, $H_\phi \in \dwt$, and that both
          graphs can be built in time $O(|\phi| \log |\phi|)$. Now, define $(e_1, \dots, e_n)$ to be 
          the edges of the form $\ledge{z'}{0}{x_i}$ for every $1 \leq i \leq n$.
	
	We claim that $\prov^{G_\phi}_{H_\phi}$ represents $\phi$
	on~$(e_1, \ldots, e_n)$.
          Calling again a \emph{possible world} a subgraph of~$H_\phi$ that
          contains all the edges not in $(e_1, \dots, e_n)$ (as these are fixed
          to~$1$), 
          it suffices to show that there is a bijection between the satisfying
          valuations $\nu$ of $\phi$ and the possible worlds of $H_\phi$ that admit a homomorphism from $G_\phi$.
          Indeed, consider the bijection defined in the obvious way: keep the edge $\ledge{z'}{T}{x_i}$ iff $X_i$ is assigned to true in the valuation.
    First, if there is a homomorphism from $G_\phi$ to a possible world $H_\phi$, then the root $z$ of the query must be mapped to~$a$ (since this is the only element with outgoing paths of length $L+1$ as prescribed by the query), and
          then it is clear that the image of any such homomorphism must take the form of a \dwt instance that contains, for each clause number $1 \leq i \leq m$, a path of length $L$ representing this clause number. This witnesses that the valuation~$\nu$ makes a variable true which satisfies clause~$i$. Hence, $\nu$ is a satisfying assignment of~$\phi$.
          Conversely, for every satisfying assignment $\nu$, considering the corresponding possible world of~$H_\phi$, we can construct a homomorphism mapping the edges of $G_\phi$ to the edges of $H_\phi$, by mapping the path of every clause to a path connected to a variable that witnesses that this clause is satisfied by~$\nu$.%
	\end{proof}
	The result then follows by an argument analogous to the one in Proposition~\ref{prop:phoml-owp-all}.
\end{proof}

\paragraph*{$\Btwp$ on $\Bdwt$.}
We then move to $\twp$ queries:
\begin{prop}\label{prop:phoml-twp-dwt}
	$\phoml(\twp, \dwt)$ does not admit an FPRAS unless $\RP=\NP$.
        This holds even for a fixed signature consisting of two labels.
\end{prop}
This result follows from a general reduction technique from $\dwt$ queries on $\dwt$
instances to $\twp$ queries on $\dwt$ instances, which allows us to conclude
using the result already shown on $\dwt$ queries (Proposition~\ref{prop:phoml-dwt-dwt}). We
note that this technique could also have been used to simplify the proofs of
hardness of exact computation in~\cite{AMS17} and~\cite{amarilli2017combined}.
We claim:
\begin{lem}\label{lem:dwt2twp}
  For any $\dwt$ query~$G$, we can compute in time~$O(\gsize{G})$ a $\twp$ query $G'$
	which is equivalent to $G$ on $\dwt$ instances: for any $\dwt$ $H$, there is a
	homomorphism from~$G$ to~$H$ iff there is a homomorphism from~$G'$ to~$H$.
\end{lem}
\begin{proof}
	Let $G$ be a $\dwt$ query. We build $G'$ following a tree traversal of~$G$.
	More precisely, we define the translation inductively as follows. If $G$
	is the trivial query with no edges, then we let the translation of~$G$ be
        the trivial query with no edges (seen as a query consisting of a
        single vertex and no edges). Otherwise, let $x$ be the root of $G$, let $\ledge{x}{R_1}{y_1},
	\ldots, \ledge{x}{R_n}{y_n}$ be the successive children, and call $G_1,
	\ldots, G_n$ the \dwt subqueries of $G$ respectively rooted at $y_1, \ldots,
	y_n$.
        Let $G_1', \ldots, G_n'$
	be the respective translations of $G_1$, \ldots, $G_n$ as
        $\twp$ queries obtained inductively. For each $1 \leq i \leq n$, let
        $G''_i$ be $\ledge{}{R_i}{G'_i} \redge{}{R_i}$, i.e., the $\twp$
        obtained by extending the $\twp$ $G_i'$ by adding an edge
        $\ledge{}{R_i}{G'_i}$ to the left (connected to the left endpoint) and
        adding an edge $\redge{}{R_i}$ to the right (connected to the right
        endpoint). In particular, if $G'_i$ is the trivial query with no edges
        then $G''_i$ is $\ledge{}{R_i}{} \redge{}{R_i}{}$. We then define
	the translation of $G$ to be the $\twp$ obtained as the
        concatenation of
        the queries $G_1'', \ldots, G_n''$, merging the
        right endpoint of $G_i''$ and the left endpoint
        of~$G_{i+1}''$ for each $1 \leq i < n$.
        This translation is
	in linear time, and the translated query has twice as many edges as the
	original query. Note that we can also inductively define a homomorphism from
	$G'$ to~$G$ mapping the first and last elements of~$G'$ to the root of~$G$: this is immediate in
	the base case, and in the inductive claim we obtain suitable homomorphisms
	from each $G_i'$ to each $G_i$ by induction and combine them in the expected
	way.
	
	We claim that, on any \dwt instance $H$, for any vertex $v$ of~$H$,
        there is a match of $G$ mapping
	the root of~$G$ to~$v$ iff there is a match of $G'$ mapping both the first variable and last
	variable of the path to~$v$. One
	direction is clear: from the homomorphism presented earlier that maps $G'$
	to~$G$, we know that any match of~$G$ in~$H$ implies that there is a match
	of~$G'$ in~$H$ mapping the first and last elements as prescribed. Let us show the converse, and let us actually show by
	structural induction on~$G$ a stronger claim: for any vertex $v$
        of~$H$, if there is a match of $G'$ mapping the
	first variable of~$G'$ to a vertex~$v$, then the last variable is also mapped
	to~$v$ and there is a match of~$G$ mapping the root variable to~$v$.
	If $G$ is the trivial query,
	then this is immediate: a match of the trivial query~$G'$ mapping the first
	variable to~$v$ must also map the last variable to~$v$ (it is the same
	variable), and we conclude. Otherwise, let us write $x$ the root of~$G$ and $\ledge{x}{R_1}{y_1},
	\ldots, \ledge{x}{R_n}{y_n}$ be the children and $G_1,
	\ldots, G_n$ the subqueries as above. We know that the match of~$G'$ maps the
	first variable to a 
	vertex $v$, and as $H$ is a \dwt instance it maps $x$ to a child $v_1$ of~$v$.
	Considering $G_1$ and its translation $G'_1$, we notice that we have a match
	of~$G_1'$ where the first variable is mapped to~$v_1$. Hence, by induction,
	the last variable is also mapped to~$v_1$, and we have
	a match of~$G_1$ where the root variable is mapped to~$v_1$. Now, 
	as $H$ is a
	\dwt instance, the next edge $\redge{}{R_1}{}$ must be mapped to
	the edge connecting $v$ and $v_1$, so that the next variable in~$G'$ is mapped
	to~$v$. Repeating this argument for the successive child edges and child queries
	in~$G$, we conclude that the last variable of~$G'$ is mapped to~$v$, and we
	obtain matches of~$G_1, \ldots, G_n$ that can be combined to a match of~$G$.
\end{proof}

Lemma~\ref{lem:dwt2twp}, when taken together with Proposition~\ref{prop:phoml-dwt-dwt}, allows us to prove Proposition~\ref{prop:phoml-twp-dwt}: it reduces in linear time
(in combined complexity) the evaluation of a \dwt query on a \dwt probabilistic instance to the evaluation of
an equivalent \twp query on the same instance and with the same signature. This establishes that any approximation
algorithm for \twp queries on \dwt instances would give an approximation for
\dwt queries on \dwt instances, which by Proposition~\ref{prop:phoml-dwt-dwt} is conditionally impossible.

These results complete Table~\ref{tab:tab}, concluding the classification
of the complexity of \phom in the labelled setting: all cases that were
intractable for exact computation are also hard to approximate, with the notable
exception of \owp queries on \dagi instances.

\section{Results in the Unlabelled Setting}\label{sec:unlabelled}

We now turn to the \textit{unlabelled} setting of probabilistic graph
homomorphism, where the signature $\sigma$
is restricted to contain only 
one label ($|\sigma| = 1$). Our results are summarized 
in Table~\ref{tab:unlabelled}: we settle all cases except $\phomul(\owp, \all)$
and $\phomul(\dwt, \all)$, for which we do not give an FPRAS or hardness of
approximation result. Note that both problems are 
\SPhard for exact computation~\cite{AMS17}. Further, they are in fact 
equivalent, because \dwt queries are equivalent to \owp
queries in the unlabelled setting (stated in~\cite{AMS17} and
reproved as Proposition~\ref{prop:phomul-dwt-dagi} below).

\paragraph*{\Bowp on \Bdagi.}
We start with \owp queries, and state the following:
\begin{restatable}{prop}{phomulowpdagi}\label{prop:phomul-owp-dagi}
	$\phomul(\owp, \dagi)$ is \SPhard already in data complexity, but it admits an FPRAS.
\end{restatable}
\noindent The positive result directly follows from the existence of an FPRAS in the
labelled setting, which we have shown in the previous section
(Proposition~\ref{prop:phoml-owp-dagi}). By contrast, the \SPhardness does not
immediately follow from previous work, as \dagi instances were not studied
in~\cite{AMS17}. We can nevertheless obtain it by inspecting the usual 
\SPhardness proof of PQE for the CQ $\exists x \, y ~ R(x), S(x, y),
  T(y)$ on TID
  instances~\cite{DBLP:series/synthesis/2011Suciu}. We give a proof in Appendix~\ref{apx:phomul-owp-dagi}.

\paragraph*{\Bdwt on \Bdagi.}
We can easily generalize the above result from \owp queries to \dwt queries,
given that they are known to be equivalent in the unlabelled setting:
\begin{restatable}[\cite{AMS17}]{propC}{phomuldwtdagi}\label{prop:phomul-dwt-dagi}
	$\phomul(\dwt, \dagi)$ is \SPhard already in data complexity, but admits an FPRAS.
\end{restatable}
\noindent This is implicit in 
\cite[Proposition 5.5]{AMS17}: we give a
self-contained proof in Appendix~\ref{apx:phomul-dwt-dagi}.

\paragraph*{\Btwp on \Bpt.}
In contrast to \owp queries, which are exactly tractable on \pt instances and
admit an FPRAS on \dagi instances, \twp queries have no FPRAS already
on \pt instances:
\begin{prop}\label{prop:phomul-twp-pt}
	$\phomul(\twp, \pt)$ does not admit an FPRAS unless $\RP=\NP$.
\end{prop}

\begin{proof}
	It suffices to prove the claim below, which is the analogue to the
        unlabelled setting of
        Claim~\ref{clm:2cnf2prov-dwt} after having transformed the query to \twp
        via Lemma~\ref{lem:dwt2twp}:
    \begin{clm}
		\label{clm:2cnftwp-ptt}
		Given a monotone 2-CNF formula $\phi$ on $n$ variables, we can build in time $O(|\phi| \log |\phi|)$ an unlabelled \twp graph $G_\phi$ and unlabelled \pt graph $H_\phi$, with the latter containing edges $(e_1, \ldots, e_n)$ such that $\prov^{G_\phi}_{H_\phi}$ represents $\phi$
		on~$(e_1, \ldots, e_n)$.
	\end{clm}
        We show this claim via a general-purpose reduction
from the labelled
        setting to the unlabelled setting, which works in
        fact when queries and instances are in the class \all. This reduction codes labels via
        specific unlabelled paths; a similar but ad-hoc technique was  used to prove~\cite[Proposition~5.6]{AMS17}:

        \begin{lem}
          \label{lem:coding}
          For any constant $k \geq 2$, given a query $G$ in the class \all and instance graph 
          $H$ in the class \all
          on a labelled signature with labels $\{1, \ldots, k\}$, we can
          construct in linear time an unlabelled query $G'$ in the class \all and instance graph
          $H'$ in the class \all such that there is a (labelled) homomorphism from $G$ to $H$ iff
          there is an (unlabelled) homomorphism from $G'$ to~$H'$. Further, if
          $G$ is a \twp then $G'$ is also a \twp, and if $H$ is a \pt then $H'$
          is also a \pt.
        \end{lem}

        \begin{proof}
          We construct $G'$ from $G$ and $H'$ from $H$ by replacing every edge
          by a fixed path that depends on the label of the edge. Specifically,
          we consider every edge $\ledge{x}{i}{y}$ of the query, where $x$ is
          the source, $y$ is the target, and $1 \leq i \leq k$ is the label.
          We code such an edge in~$G'$ by a path defined as follows: $x \rightarrow^{k+3} ~ \leftarrow ~ \rightarrow^{i+1} ~ \leftarrow^{k+2} ~ y$, where exponents denote repeated edges and where intermediate vertices are omitted. We code the instance $H$ to~$H'$ in the same way. This process is clearly linear-time, and it is clear that if $G$ is a \twp then $G'$ is also a \twp, and that if $H$ is a \pt then $H'$ is also a \pt. Further, to establish correctness of the reduction, one direction of the equivalence is trivial: a homomorphism $h$ from~$G$ to~$H$ clearly defines a homomorphism from~$G'$ to~$H'$ by mapping the coding in~$G'$ of every edge $e$ of~$G$ to the coding of the image of~$e$ by~$h$ in~$H'$

          What is interesting is the converse direction of the equivalence. We
          show it via a claim on homomorphic images of the coding of individual edges:
          for any $1 \leq i \leq k$, letting $e'$ be the coding of an edge $e = \ledge{x}{i}{y}$,
          for any homomorphism $h'$ from $e'$ to $H'$, there must exist an edge $f = \ledge{a}{i}{b}$ in~$H$ such that $h'$ maps~$x$ to~$a$ and $y$ to~$b$. This claim implies the converse direction of the equivalence: if there is a homomorphism $h'$ from~$G'$ to~$H'$, then applying the claim to the restrictions of~$h'$ to the coding of each edge of~$G$, we see that $h'$ defines a function~$h$ that maps the vertices of~$G$ to vertices of~$H$, and that $h$ is a homomorphism. Hence, all that remains is to prove the claim, which we do in the rest of the proof.

          Consider an edge $e = \ledge{x}{i}{y}$ as in the claim statement, and let $e'$ be its coding and $h'$ the homomorphism mapping $e'$ to~$H'$. Observe that, in~$H'$, the only directed paths of length $k+3$ are the first $k+3$ edges of the coding of edges of~$H$. (This hinges on the fact that the paths of length $k+3$ defined in the coding of edges of~$H$ are never adjacent in~$H'$ to another edge that goes in the same direction, even across multiple edges, and no matter the directions of edges in~$H$.) 
          This means that, considering the directed path $\rightarrow^{k+3}$ at the beginning of~$e'$, there must be an edge $f = \ledge{a}{j}{b}$ of~$H$, with coding~$f'$ in~$H'$, such that the source~$x$ of~$e$ is mapped to the source~$a$ of~$f$, and the first $k+3$ edges of~$e'$ are mapped to the first $k+3$ edges of~$f'$. What remains to be shown is that $i=j$ and that $y$ is mapped to~$b$.

          To this end, we continue studying what can be the image of~$e'$ into~$f'$. After the directed path $\rightarrow^{k+3}$, the
          next edge $\leftarrow$ of~$e'$ must have been mapped forward to the next edge $\leftarrow$ of~$f'$: indeed, it cannot be mapped backwards on the last edge of the preceding path $\rightarrow^{k+3}$ because $k+3 > 1$ and $i+1 > 1$ so the next edges $\rightarrow^{i+1}$ would then have no image. Then the next directed path $\rightarrow^{i+1}$ of~$e'$ is mapped in~$f'$, necessarily forward because we fail if we map the first edge backwards: this implies that there at least as many edges going in that direction in~$f'$ as there are in~$e'$, i.e., $i \leq j$.
          Now, the last path $\leftarrow^{k+2}$ of~$e'$ cannot be mapped backwards because $k+2 > i+1$, so we must map it forwards in~$f'$: for this to be possible, we must have reached the end of the directed path $\rightarrow^{j+1}$ in~$f'$, so that we have $j = i$. We are now done reading~$e'$ and $f'$, so we have indeed mapped $y$ to~$b$. This, along with $i=j$, establishes that the claim is true, and concludes the proof.
        \end{proof}
        We can thus prove Claim~\ref{clm:2cnftwp-ptt}, starting from
        Claim~\ref{clm:2cnf2prov-dwt} and translating the labelled \dwt query first to a labelled \twp query via
        Lemma~\ref{lem:dwt2twp}, and then the labelled \twp query and \pt instance (with precisely two labels) to an \textit{unlabelled} \twp query and \pt instance via Lemma~\ref{lem:coding}. Using the same argument as in Proposition~\ref{prop:phoml-owp-all}, we conclude the proof of Proposition~\ref{prop:phomul-twp-pt}.
\end{proof}

\section{DNNF Lower Bounds}\label{sec:dnnf-bounds}
In this section, we investigate how to represent the provenance of
the query-instance pairs that we consider. More specifically, we study
whether there exist polynomially-sized representations in tractable circuit
classes of Boolean provenance functions $\prov^G_H$, for $G \in \Gmc$ and $H \in
\Hmc$ in the graph classes studied in this paper. Certainly, for every graph class $\Gmc$ and $\Hmc$, the (conditional) non-existence of an FPRAS for $\phom(\Gmc, \Hmc)$ implies that, conditionally, we cannot compute nOBDD representations of provenance in polynomial time combined complexity---as otherwise we could obtain an FPRAS via Theorem~\ref{thm:nobdd-weighted-count}.
In fact, beyond nOBDDs, it follows from~\cite[Theorem~6.3]{DBLP:conf/stoc/ArenasCJR21a} that, conditionally, we cannot tractably compute provenance representations even in the more general class of \emph{structured DNNFs}.
Indeed, as for nOBDDs, fixed edges in the reductions can be handled by conditioning~\cite[Proposition~4]{PD08}.%

However, even in settings where there is conditionally no combined FPRAS, it could be the case
that there are polynomial-\emph{sized} tractable circuits that are difficult to compute, or that we can tractably compute circuits in a more general formalism such as \emph{unstructured} DNNF circuits.
The goal of this section is to give a negative answer to these two questions, for all of the non-approximable query-instance class pairs studied in Sections~\ref{sec:labelled} and~\ref{sec:unlabelled}.

Specifically, we show lower bounds on the size of DNNF circuits for infinite families of graphs taken from these classes. Remember that 
DNNF is arguably the most 
general knowledge compilation circuit class
that still enjoys some tractable properties~\cite{DM02}. Hence, these lower bounds imply that no tractable provenance representation exists in other tractable subclasses of DNNFs, e.g., structured DNNFs~\cite{PD08}, or Decision-DNNFs~\cite{beame2017exact}.
We also emphasize that, unlike the inapproximability results of
Sections~\ref{sec:labelled} and~\ref{sec:unlabelled} which assumed $\RP \neq \NP$, 
all of the DNNF lower bounds given here are unconditional.

We first show a \textit{strongly} exponential lower bound for labelled $\owp$ 
queries on instances in the class \all:
\begin{prop}\label{prop:owp-all-lb}
  Let $\sigma$ be any fixed signature containing at least two labels.
	There is an infinite family $G_1, G_2, \ldots$ of labelled $\owp$ queries and an
	infinite family $H_1, H_2, \ldots$ of labelled instances in the class
        \all and on signature~$\sigma$ such that, for any $i
	> 0$, any DNNF circuit representing the Boolean function $\prov^{G_i}_{H_i}$ has
        size $2^{\Omega(\gsize{G_i} + \gsize{H_i})}$.
\end{prop}
\begin{proof}
	By \emph{treewidth} of a monotone 2-CNF formula, we mean the treewidth of the
	graph on the variables whose edges correspond to clauses in the expected way;
	and by \emph{degree} we mean the maximal number of clauses in which any
	variable occurs.
	Let us consider an infinite family $\phi_1, \phi_2, \ldots$ of monotone 2-CNF
	formulas of constant degree $d=3$ whose treewidth is linear in their size:
	this exists by~\cite[Proposition~1, Theorem~5]{GroheM09}. 
	We accordingly know by \cite[Corollary~8.5]{ACMS20} that any DNNF computing $\phi_i$ must
	have size $2^{\Omega(|\phi_i|)}$ for all~$i > 1$.
	Using 
	Claim~\ref{clm:2cnf2prov}, we obtain infinite families $G_1, G_2, \ldots$ of
	\owp and $H_1, H_2, \ldots$ of graphs in the class \all and on signature~$\sigma$ such that $\prov^{G_i}_{H_i}$
        represents $\phi_i$ on some choice of edges, and we have $\gsize{G_i} + \gsize{H_i} =
	O(|\phi_i|)$ for all $i>0$ (from the running time bound). Now, any representation of
	$\provl^{G_i}_{H_i}$ as a DNNF can be translated in linear time to a
	representation of $\phi_i$ as a DNNF of the same size, simply by renaming the edges $(e_1,
	\ldots, e_n)$ to the right variables, and replacing all other variables by the
	constant~$1$. This means that the lower bound on the size of DNNFs
	computing $\phi_i$ also applies to DNNFs representing $\prov^{G_i}_{H_i}$,
	\ie, they must have size at least $2^{\Omega(|\phi_i|)}$, hence
        $2^{\Omega(\gsize{G_i} + \gsize{H_i})}$ as we claimed.
\end{proof}
We now present lower bounds for the remaining non-approximable query-instance class
pairs, which are not strongly exponential but rather \textit{moderately} exponential. This is because our
encoding of CNFs into these classes (specifically,
Claim~\ref{clm:2cnf2prov-dwt}, and its images by Lemma~\ref{lem:dwt2twp} and
Lemma~\ref{lem:coding})
give a bound which is not linear but linearithmic (i.e., in $O(|\phi| \log
|\phi|)$).
We leave to future work the question of proving strongly exponential lower bounds for
these classes,
like we did in Proposition~\ref{prop:owp-all-lb}.
\begin{prop}\label{prop:twp-dwt-lb}
  Let $\sigma$ be any fixed signature containing at least two labels.
	For any $\epsilon > 0$, there is an infinite family $G_1, G_2, \ldots$
        of labelled \dwt queries and an
	infinite family $H_1, H_2, \ldots$ of labelled \dwt instances on signature~$\sigma$ such that, for any $i
	> 0$, any DNNF circuit representing the Boolean function $\prov^{G_i}_{H_i}$ has
        size at least $2^{\Omega \left( \left(\gsize{G_i} + \gsize{H_i} \right)^{1-\epsilon} \right)}$.
\end{prop}
\begin{proof}
	The proof is identical to that of Proposition~\ref{prop:owp-all-lb},
        except that we apply Claim~\ref{clm:2cnf2prov-dwt}: for all $i > 0$, $\gsize{G_i} + \gsize{H_i} = O(|\phi_i| \log{|\phi_i|})$.
        We perform a change of variables:
        if we write $y = |\phi_i| \log |\phi_i|$, then
        we can show that $|\phi_i| = e^{W(y)}$, where $W$ denotes the Lambert
        $W$ function~\cite{CGHJK96}; equivalently $|\phi_i| = y/W(y)$
        as the $W$ function satisfies $W(z) e^{W(z)} = z$ for all $z > 0$.
        Thus, the lower bound of
        $2^{\Omega(|\phi_i|)}$ on DNNF representations of~$\phi_i$
        implies that any DNNF for $\prov^{G_j}_{H_j}$ has size at least
        $2^{\Omega \left( \frac{\gsize{G_i} + \gsize{H_i}}{W(\gsize{G_i} + \gsize{H_i})} \right) }$.
        In particular, as $W$ grows more slowly than $n^\epsilon$ for any
        $\epsilon>0$, this gives a bound of $2^{\Omega \left( \left(\gsize{G_i} + \gsize{H_i} \right)^{1-\epsilon} \right)}$
        for sufficiently large $\phi_j$.
\end{proof}
The proof for the following two claims are analogous to that of
Proposition~\ref{prop:twp-dwt-lb}, but using Lemma~\ref{lem:dwt2twp} (for the
first result) and Claim~\ref{clm:2cnftwp-ptt} (for the second result):
\begin{prop}
  Let $\sigma$ be any fixed signature containing at least two labels.
	For any $\epsilon > 0$, there is an infinite family $G_1, G_2, \ldots$
        of labelled \twp queries and an
	infinite family $H_1, H_2, \ldots$ of labelled \dwt instances on signature~$\sigma$ such that, for any $i
	> 0$, any DNNF circuit representing the Boolean function $\prov^{G_i}_{H_i}$ has
        size at least $2^{\Omega \left( \left(\gsize{G_i} + \gsize{H_i} \right)^{1-\epsilon} \right)}$.
\end{prop}

\begin{prop}
  Let $\sigma$ be any fixed signature containing at least two labels.
	For any $\epsilon > 0$, there is an infinite family $G_1, G_2, \ldots$ of unlabelled \twp queries and an
	infinite family $H_1, H_2, \ldots$ of unlabelled \pt instances on signature~$\sigma$ such that,
        for any $i > 0$,
        any DNNF circuit representing the Boolean function $\prov^{G_i}_{H_i}$
        has
        size at least $2^{\Omega \left( \left(\gsize{G_i} + \gsize{H_i} \right)^{1-\epsilon} \right)}$.
\end{prop}
We finish by remarking that all of the lower bounds above apply to acyclic query classes (\ie, queries of treewidth 1), for which non-probabilistic query evaluation is well-known to be linear in combined complexity~\cite{Y81}. Thus, these results give an interesting example of query classes for which query evaluation is in linear-time combined complexity, but
computing even a DNNF representation of query provenance is exponential (as we
presented in Result~\ref{res:provenance}).

\section{Consequences}\label{sec:applications}
In this section, we consider some corollaries and extensions to the results above.%

\paragraph*{Optimality of a Previous Result.}
Recall from the introduction that, as was shown in~\cite{vBM23}, 
PQE for self-join-free conjunctive queries of bounded hypertree width
admits a combined FPRAS (in the general setting of probabilistic databases,
rather than probabilistic graphs):
\combinedsjfcq*
\noindent Can a stronger result be achieved? Our
Proposition~\ref{prop:phomul-twp-pt} immediately implies the following:
\begin{cor}\label{cor:tw1tw1}
	Assuming $\RP \neq \NP$,
	even on a fixed signature consisting of a single binary relation
	there is no FPRAS to approximate the probability of an input treewidth-1 CQ on an input treewidth-1 TID instance.
\end{cor}
\noindent Hence, tractability no longer holds 
with self-joins. So, as unbounded
hypertree width queries are intractable in combined complexity even for 
\textit{deterministic} query evaluation, 
we have:
\begin{cor}\label{cor:optimality}
	Assuming $\RP \neq \NP$, the result in Theorem~\ref{thm:combinedsjfcq} is 
	optimal in the following sense: relaxing \textit{either} the self-join-free 
	or bounded-hypertree-width condition on the query implies the non-existence 
	of a combined FPRAS.
\end{cor}

\paragraph*{Network Reliability.}
The \emph{two-terminal network reliability problem}, dubbed \emph{ST-CON} for
brevity,
intuitively asks 
the following: given a directed graph with probabilistic edges and 
with
source and target vertices $s$ and $t$, compute the probability that $s$ and $t$ remain connected,
assuming independence across edges.
Formally, working in the unlabelled setting of a signature $\sigma$ with
$|\sigma| = 1$, we are given a probabilistic graph $(H, \pi)$ on
signature~$\sigma$ together with two vertices $s$ and $t$.
We must compute
the probability $\sum_{H' \subseteq H\ \text{s.t.\ there is an } st-\text{path
in } H'} \Pr_\pi(H')$,
where by \emph{$st$-path} we mean a directed path from~$s$ to~$t$.
Valiant showed that ST-CON is $\SP$-complete~\cite[Theorem~1]{V79},
and
Provan and Ball showed that 
this holds already 
on directed acyclic graphs~\cite[Table~2]{PB83}. 
Hardness also holds for the related problem of \emph{all-terminal reliability}~\cite[Table~1]{PB83},
which asks 
for the probability that the probabilistic graph remains connected as a whole.
Given the inherent \SPhardness of these problems, 
subsequent research has focused on developing tractable approximations.

Although significant progress has been made 
on FPRASes for the related problems of \emph{all-terminal
(un)reliability}~\cite{GJ19,K01}, 
designing an FPRAS for ST-CON has remained open. This question was even open for the restricted case of directed acyclic graphs;
indeed, it was explicitly posed as an open problem by Zenklusen and Laumanns~\cite{ZL11}.
We now point out that the nOBDD construction of
Proposition~\ref{prop:phoml-owp-dagi} implies an FPRAS for ST-CON on DAGs, again by leveraging the approximate counting result of Arenas et al.~\cite{ACJR21}:

\begin{thm}\label{thm:twoterminal-dag-fpras}
	There exists an FPRAS for the ST-CON problem over directed acyclic graphs.
\end{thm}
\begin{proof}
  Given as input an unlabelled probabilistic \dagi instance $H = (V, E)$ and two distinguished source and target vertices $s$ and $t \in V$, construct as follows the labelled \dagi instance $H' = (V, E, \lambda)$, whose set of labels is $\{R, R_s, R_t, R'\}$.
  All vertices and edges are identical to that of $H$, but every edge of the form $(s, x)$ emanating from $s$ is assigned label $\lambda((s,x)) = R_s$, every edge $(x, t)$ directed towards $t$ is assigned label $\lambda((x, t)) = R_t$, and every other edge $(x, y)$ is assigned the label $\lambda((x,y)) = R$. In the case that $(s,t) \in E$, then assign $\lambda((s,t)) = R'$.
	
	Now, by the result in Proposition~\ref{prop:phoml-owp-dagi}, we can construct an nOBDD for each of the following $|E|$ different labelled \owp queries: $\ledge{}{R'}{}$, $\ledge{}{R_s}{} \ledge{}{R_t}{}$, $\ledge{}{R_s}{} \ledge{}{R}{} \ledge{}{R_t}{}$, \dots, $\ledge{}{R_s}{} \left( \ledge{}{R}{} \right)^{|E|-2} \ledge{}{R_t}{}$. All of the nOBDDs have the same ordering (given by some topological ordering of the edges of $H'$), so we may take their disjunction to obtain a (complete) nOBDD $D$ in linear time, whose accepting paths are in bijection with the $(s,t)$-connected valuations of the edges in $H$. From here 
        we conclude by applying Theorem~\ref{thm:nobdd-weighted-count}.
\end{proof}
\noindent We remark that an improved running time bound for 
Theorem~\ref{thm:twoterminal-dag-fpras} was 
obtained independently by Feng and Guo~\cite{FG24}.
It may also be possible to improve the running time bounds obtained for our approach by
leveraging recent faster algorithms for nOBDD approximate
counting such as~\cite{meel2024faster}, which give improved bounds for Theorem~\ref{thm:nobdd-unweighted-count}.

\section{Regular Path Queries}\label{sec:rpqs}
In this section, we investigate the approximability of \textit{regular path 
queries} (RPQs) on probabilistic graphs, measured in \textit{data complexity}. 
Note that this differs from the rest of the paper in two respects. First, 
RPQs are distinct in expressiveness as a query class from 
conjunctive queries (CQs) studied earlier: they are generally not even expressible
as \textit{unions} of conjunctive queries (UCQs). Second, we focus on data 
complexity (\ie, 
treating the query as fixed, and the probabilistic graph as the sole input) in 
this section. This is in contrast to the combined complexity setting considered 
elsewhere in the paper; data complexity is relevant for RPQs because computing their
probability on a probabilistic graph is not necessarily approximable,
unlike CQs and UCQs which are always approximable in data complexity. It 
turns out that many of the techniques and results 
from earlier in the paper have direct applications for approximability in 
data complexity of RPQs, motivating the 
shift of focus here.

\paragraph*{Preliminaries.}
We briefly review some background on languages, automata, and RPQs relevant 
to this section. An 
\textit{alphabet} $\Sigma$ is some finite set of symbols, which we typically 
denote with lowercase letters, \ie, $a, b, c, \dots$ A \textit{word} 
over $\Sigma$ is a finite sequence of symbols from $\Sigma$. We denote by 
$\Sigma^*$ the (infinite) set of all possible words over $\Sigma$, and call a 
subset $L \subseteq \Sigma^*$ of words a \textit{language} over $\Sigma$.

We define 
a
\emph{deterministic finite automaton} (DFA) over an alphabet $\Sigma$ as a 
tuple 
$A
= (Q, \Sigma, q_0, F, \delta)$  where $Q$ is a finite set of \emph{states},
$\Sigma$ is the alphabet, $q_0 \in Q$
is the \emph{initial state}, $F \subseteq Q$ are the \emph{final states}, and
$\delta\colon Q \times \Sigma \to Q$ is the transition function. Let $w = 
a_1a_2\dots a_n$ be a word over the alphabet $\Sigma$. We say that $A$ 
\textit{accepts} $w$ if there exists some sequence of states $r_0, r_1, 
\dots, r_n$ in $Q$ such that: (1.) $r_0 = q_0$, (2.) $r_{i+1} = \delta(r_i, 
a_{i+1})$ for all $0 \leq i \leq n -1$, and (3.) $r_n \in F$. The set of 
strings accepted by $A$ forms a language over $\Sigma$, which we call the 
language \textit{recognized} by $A$. We call a language \textit{regular} if 
it is accepted by some DFA. We often also use a regular expression syntax 
as shorthand for specifying regular languages: for instance, $ab^* c^+$ 
denotes the language of the words that start with $a$, continue with zero or
more $b$'s, and then finish by one or more $c$'s.

We define a \textit{regular path query} (RPQ) $Q$ on a regular language $L$ on
alphabet~$\Sigma$, 
denoted $Q = \rpq(L)$, as a Boolean query whose semantics are given as 
follows: a labelled graph $H$ with signature $\sigma := \Sigma$
\emph{satisfies} $\rpq(L)$ iff there exists some 
\owp query $G = {\ledge{}{a_1}{} \cdots \ledge{}{a_n}{}}$ such 
that $a_1 \dots a_n \in L$ and $G \homo H$.

\begin{exa}
  For singleton languages $L_w = \{w\}$ containing precisely one word~$w$, then
  $\rpq(L_w)$ is equivalent to the \owp query formed from~$w$. As another
  example, when we take $L_0 = a b^* c$, then $\rpq(L_0)$ holds on the labelled 
  graphs
  such that there is a walk (i.e., a path which is not necessarily
  simple) which consists of $b$-edges and goes from an $a$-edge to a $c$-edge.
\end{exa}

Clearly, the representation of 
$Q$ is determined by the representation of the underlying regular language 
$L$, and so for concreteness we may assume throughout this section that $L$ 
is specified as a DFA. However, as we use the data complexity perspective, it does not matter much: all results continue to hold if, for example, 
$L$ is specified as a regular expression.

We study the \emph{probabilistic query evaluation} problem for RPQs in
\emph{data complexity}, i.e., for each
RPQ $Q = \rpq(L)$ on alphabet~$\Sigma$, we study a problem $\pqe{Q}$ defined
as follows. The input to $\pqe{Q}$ is a probabilistic graph $(H, \pi)$ on the 
set of labels $\sigma := \Sigma$, and $\pqe{Q}$ asks for the total probability
of the subgraphs $H' \subseteq H$ of~$(H, \pi)$ that satisfy $Q$ according to
the definition above.
Note that fixing the query~$Q$ fixes in particular the signature~$\sigma$ over which input instances are expressed.
In the statement of our upper bounds, we will also
consider the variant of $\pqe{Q}$ where the input is a probabilistic arity-two database
instead of a probabilistic graph, and mention it explicitly when stating the
upper bounds.

Some of our results on RPQ will use the notion of \emph{Boolean provenance} for
RPQs, which is defined in the same way as for the queries that we studied 
earlier in the paper. Recalling the notion of a \emph{valuation} of a 
graph~$H$, fixing an RPQ $Q$,
the \emph{provenance} of~$Q$ on~$H$ is the Boolean function
$\prov^Q_H$ having as variables the edges~$E$ of~$H$ and mapping
every valuation $\nu$ of~$E$ 
to~$1$ (true) or $0$ (false) depending on whether $H_\nu$ satisfies~$Q$ or not.
Like Definition~\ref{def:provrepresents}, given a 
Boolean formula $\phi$ whose variables $\{e_1, \ldots, e_n\} \subseteq E$ are 
edges of $H$, we say that $\prov_H^Q$ \emph{represents} $\phi$ on $(e_1,...,e_n)$ if 
for every valuation $\nu: E \rightarrow \{0,1\}$ that maps edges not in 
$\{e_1,...,e_n\}$ to $1$, we have $\nu\models\phi$ if and only if 
$\prov_H^Q(\nu)=1$.

\paragraph*{Infix-Free Languages.}
When evaluating RPQs on graphs, we have no fixed ``endpoints'' constraining 
where a query match must begin and end. Accordingly, 
distinct languages may be equivalent, in the sense that they give rise to the 
same RPQ. 
For example, we have that $\rpq(a) = \rpq(aa^*)$, since any labelled graph that 
contains a match of $aa^*$ also contains a match of $a$, and vice versa. To 
formalize this idea and relate these languages, we assume that every language 
in this section is \textit{infix-free}, in a sense that we define below.

	Let $L$ be a regular language over some alphabet $\Sigma$.
        Define a partial order 
	$\preccurlyeq_L$ over words of $L$ as follows: for words $u, v \in L$, we 
	have $u \preccurlyeq_L v$ iff there exists words $s, t \in \Sigma^*$ such 
	that $v 
	= 
	sut$. We call $u$ an \textit{infix} of $v$, and a \textit{strict infix} 
	of $v$ if additionally $u \neq v$. The 
	\textit{infix-free sublanguage} of $L$, denoted $\minl(L)$, is the 
	(possibly 
	infinite) set of minimal elements of $\preccurlyeq_L$. We further say 
	that a language $L$ is \textit{infix-free} if $L = \minl(L)$.
Note that the infix-free sublanguage of an infinite language may be finite; 
recalling the example above, $\minl(aa^*) = a$. The following proposition 
establishes that the infix-free sublanguage of any regular language is itself 
regular: it immediately follows from the observation that $\minl(L) = L 
\setminus ((\Sigma^+ L \Sigma^*) \cup (\Sigma^* L \Sigma^+)) = L \cap 
((\Sigma^+ L \Sigma^*) \cup (\Sigma^* L \Sigma^+))^C$, given that regular
languages are closed under union, intersection, concatenation, and 
complementation.
\begin{propC}[{\cite[Proposition~6]{PR10}}]
	Let $L$ be a regular language. Then $\minl(L)$ is regular.
\end{propC}
Moreover, a regular language gives rise to the same RPQ as that of  
its infix-free sublanguage.
\begin{prop}
	Let $L$ be a regular language. Then $\rpq(L) = \rpq(\minl(L))$.
\end{prop}
\begin{proof}
	Consider an input graph $G$ with signature $\Sigma$. We show that 
	$\rpq(L)$ evaluates to true iff 
	$\rpq(\minl(L))$ evaluates to true. We first show the forward direction; 
	assume that $\rpq(L)$ evaluates to true on $G$, and let $w \in L$ be a 
	word labelling some corresponding walk $\pi$ in $G$. By the construction of 
	$\minl(L)$, there exists $v \in \minl(L)$ such that there exist 
	(possibly empty) words $s, t \in \Sigma^*$ so that $w = svt$. Thus, we 
	must also 
	have a match from the word $v \in \minl(L)$ on $G$, and so 
	$\rpq(\minl(L))$ evaluates to true: just consider the 
	appropriate subgraph of $\pi$ witnessing a match of $w$. For the 
	converse, suppose that $\rpq(\minl(L))$ evaluates to true on $G$, and let 
	$w \in \minl(L)$ be a word labelling some corresponding
        walk in $G$. Then 
	since $\minl(L) \subseteq L$, we have that $w \in L$ and so $\rpq(L)$ 
	evaluates to true as well.
\end{proof}

Let $L$ be an infix-free regular language. We call $\rpq(L)$ \textit{bounded} if 
$L$ is finite, and \textit{unbounded} otherwise.

\paragraph*{Bounded RPQs.}
We observe that PQE can be tractably approximated for \textit{every} bounded
RPQ, because bounded RPQs are just a restricted kind of union of conjunctive queries 
(UCQs) on
binary signatures, for which PQE is known to admit an 
FPRAS~\cite{DBLP:series/synthesis/2011Suciu}. This uses
the Karp-Luby FPRAS (essentially a special case of the result in 
Theorem~\ref{thm:nobdd-weighted-count}):

\begin{prop}
  \label{prp:finiterpqfpras}
	Let $Q$ be a bounded RPQ. Then $\pqe{Q}$ admits an FPRAS.
        This holds even if the input instance is a probabilistic arity-two database.
\end{prop}
\begin{proof}
	By \cite[Proposition~5, Theorem~11]{BF019}, $Q$ is equivalent to a UCQ
        which can be computed
	in time independent of the instance graph size, i.e., in constant time
        in data complexity.
        We can then compute the provenance of the fixed query $Q$ on the input
        graph, and obtain it as a disjunctive normal form (DNF) formula, in polynomial-time data
        complexity.
        We may then apply the Karp-Luby FPRAS 
        to compute the probability of this 
	UCQ~\cite[Section~5.3.2]{DBLP:series/synthesis/2011Suciu}.
\end{proof}

We further note that, for some bounded RPQs, we can even solve $\pqe{Q}$ 
exactly
in $\PTIME$: for instance for the language $ab$. By contrast, there are other
bounded RPQs, for example the language $abc$, for which the problem is 
\SPhard.
In fact, there is a dichotomy for
$\pqe{Q}$ over the 
bounded 
RPQs $Q$ between $\PTIME$ and $\SP$: this follows from the more general
Dalvi-Suciu dichotomy on PQE for UCQs~\cite{DBLP:journals/jacm/DalviS12}.

\paragraph*{Unbounded RPQs.}
For unbounded RPQs, we know that exact PQE is always \SPhard, by the
results of~\cite{amarilli2022dichotomy}. Indeed, recall that a query is 
said to be \textit{closed under homomorphisms} if the following holds: if a 
graph $H$ satisfies the query and $H \homo H_0$, then $H_0$ also 
satisfies the query. The following is easy to observe:

\begin{prop}
  Every RPQ is closed under homomorphisms.
\end{prop}

\begin{proof}
  Let $Q = \rpq(L)$ be an RPQ. Suppose $H$ satisfies $Q$. Then by definition, 
  there exists some \owp instance $G = {\ledge{}{a_1}{} \cdots \ledge{}{a_n}{}}$ 
  such that $a_1 \dots a_n \in L$ and $G \homo H$. Now suppose $H \homo H_0$. 
  Since the composition of two homomorphisms is itself a homomorphism, we 
  have that $G \homo H_0$, and so $H_0$ satisfies $Q$ as required.
\end{proof}

Hence, PQE for unbounded RPQs is \SPhard, because PQE is \SPhard in 
data complexity for every unbounded homomorphism-closed query on an arity-two
signature~\cite[Theorem~3.3]{amarilli2022dichotomy}.%

In terms of approximation, we now show that the PQE problem for unbounded
RPQs is always at 
least as hard as the two-terminal network reliability problem (ST-CON) in 
directed graphs:

\begin{prop}
  \label{prp:stcon2rpq}
	Let $Q$ be any fixed unbounded RPQ on alphabet~$\Sigma$. There is a 
	polynomial-time reduction from
    ST-CON to PQE
    for~$Q$. Formally, given an unlabelled graph $(H, \pi)$
    with probabilistic edges and given vertices $s$ and~$t$ of~$H$,
    we can build in polynomial time a labelled graph $(H', \pi')$ with
    probabilistic edges on signature $\Sigma$ and vertices $s'$ and $t'$
    of~$H'$ such that the answer to ST-CON on $(H, \pi)$ is the same as 
    the answer to $\pqe{Q}$ on~$(H', \pi')$.
\end{prop}
\begin{proof}
	Let $Q = \rpq(L)$ be the unbounded RPQ. We may assume without loss of 
	generality that $L$ is infix-free (for if not, we can simply make it
        infix-free,
        and the answer to $\pqe{Q}$ is unchanged). We 
	build our reduction from ST-CON on the unlabelled probabilistic graph 
	$(H, \pi)$, with distinguished vertices $s$ and $t$.
	
	By the pumping lemma, there is a word $xyz \in L$ 
	comprising the concatenation of three subwords $x$, $y$, and $z$ such 
	that $xy^nz \in L$ for all $n \geq 0$, and $y \neq \epsilon$. We now show 
	that $x \neq \epsilon$ and $z \neq \epsilon$. Indeed, suppose 
	for a contradiction that $x = \epsilon$: then we have $x y^0 z = z \in L$ and 
	so $yz \not\in L$ (since $L$ is infix-free), a contradiction. Similarly, if 
	$z = \epsilon$, we have $x \in L$ and so $xy \not\in L$, again a 
	contradiction.
	
	Now, we construct a labelled probabilistic graph $(H', \pi')$ from the 
	unlabelled ST-CON instance graph $(H, \pi)$ as follows. $H'$ is identical 
	in structure to $H$, except for the addition of a new vertex $x_i$, 
	connected to~$s$ via an $x$-labelled path whose edges all have probability 1,
        and similarly we add a new path connecting~$t$ to
        a new vertex $x_e$ via a $z$-labelled path whose edges all have
        probability 1.
        Furthermore, the original unlabelled 
	edges of $H'$ are each replaced by an $y$-labelled path, with the first edge 
	of the path carrying the same probability as the original corresponding 
        edge in $H$, and each remaining edge of the path (if any) carrying probability 1.
	
	It remains to show correctness, i.e., that the answer to ST-CON on $(H, 
	\pi)$ is the same as the answer to $\pqe{Q}$ on~$(H', \pi')$. Indeed, 
	consider the probability-preserving bijection between subgraphs of $(H, 
	\pi)$ and $(H', \pi')$ defined in the natural way: keep an edge in $H$ 
	iff its 
	corresponding edge in $H'$ is kept (preserving all the additional edges 
	of probability 1 that arose from the construction of $H'$). It is clear 
	that for every deterministic subgraph $H_d \subseteq H$ such that $t$ is 
	reachable from 
	$s$, its counterpart $H'_d \subseteq H'$ contains a path labelled $xy^nz$ 
	for some $n \geq 0$, and thus $H'_d$ satisfies $Q$ as desired. We now 
	show that for any $H_d \subseteq H$ such that $t$ is \textit{not} 
	reachable from $s$, its counterpart $H'_d \subseteq H'$ fails to satisfy 
	$Q$. It is easy to see that the label for any path in $H'$ forms an infix 
	of $xy^nz$ for some $n \geq 0$. Thus, the label for any path in $H'_d$ 
	is a strict infix of $xy^nz$, since if the path was labelled $xy^nz$, 
	there would be a path from $x_i$ to $x_e$ in $H'_d$, and thus from $s$ to 
	$t$ in $H_d$. Finally, as $L$ is infix-free, and $xy^nz \in L$, no strict 
	infix of $xy^nz$ can belong to $L$. This implies that no path in $H'_d$ 
	is labelled with a word of $L$, and so $H'_d$ does not satisfy $Q$.
\end{proof}

This result implies that finding an FPRAS for $\pqe{Q}$ for \emph{even one} 
unbounded RPQ would imply the existence of an FPRAS for ST-CON, solving a 
long-standing open problem. It is also easy to see that, for some specific 
unbounded RPQs $Q$, the problem $\pqe{Q}$ is in fact polynomial-time 
equivalent to ST-CON.

\begin{exa}
  \label{exa:stconrpq}
  Consider the RPQ $Q$ for the (infix-free) language $a b^* c$. Then $Q$ is
  unbounded, so by Proposition~\ref{prp:stcon2rpq} we know that ST-CON 
  reduces in
  polynomial-time to $\pqe{Q}$. However, the converse is also true: given a
  probabilistic graph $(H, \pi)$ on $\Sigma = \{a, b, c\}$, we can build an
  unlabelled probabilistic graph $(H', \pi')$ in the following way.
  Initialize $H'$ as the graph of the $b$-edges of~$H$. Now, add to~$H'$ a 
  source $s$ and target
  $t$. Then connect $s$ with an edge in~$H'$ to the target of each $a$-edge $e$ in~$H$
  with probability $\pi(e)$, and likewise connect the source of each $c$-edge
  $e$ in~$H$ with an edge to~$t$ in~$H'$ with probability $\pi(e)$. It is now
  easy to see that the answer to ST-CON on~$(H', \pi')$ is the same as that
  of~$\pqe{Q}$ on~$(H, \pi)$.
\end{exa}

Analogously to this example, we can show that $\pqe{Q}$ is equivalent to ST-CON
for any unbounded RPQ $Q$ whose infix-free language is a so-called \emph{local
language}\footnote{This connection was realized in an independent collaboration with Gatterbauer, Makhija,
and Monet~\cite{amarilli2025resilience}, in which similar techniques are used to solve a different 
problem.}. The class of \emph{local languages}, which includes for instance
$ab^*c$ from the previous example, is a class of regular languages admitting
several equivalent characterizations. In particular, a language is local iff it
is recognized by a so-called \emph{local automaton}. Recall the definition of 
a deterministic finite automaton (DFA) over an alphabet $\Sigma$ as a tuple $A
= (Q, \Sigma, q_0, F, \delta)$. 
A DFA is called \emph{local} if all transitions labelled with the same letter 
lead to
the same state; formally for each letter $a \in \Sigma$, there exists a 
unique state $q_a \in Q$ such that $\delta(q, a) = q_a$ for each $q \in Q$.

We claim:

\begin{prop}
  \label{prp:rpq2stcon}
	Let $Q$ be any fixed unbounded RPQ on alphabet~$\Sigma$ whose infix-free
        language $L$ is local. Then there is a polynomial-time reduction from
        $\pqe{Q}$ to ST-CON.
        This holds even if the input instance is a probabilistic arity-two database.
\end{prop}
\begin{proof}
  Let $A$ be a local DFA for~$L$.

  We build an unlabelled probabilistic graph $(H', \pi')$ from the input
  probabilistic arity-two database $(H, \pi)$ with alphabet $\Sigma$, and
  show that the answer to $\pqe{Q}$ on~$(H, \pi)$ is the same as the answer of
  ST-CON on~$(H', \pi')$.
  Let $V$ be 
  the active domain of~$H$; the active domain of $H'$ will be a subset of 
  $(V \times \Sigma) \sqcup (V \times Q)$.
  For each letter $a \in \Sigma$ and state $r$ of~$A$ such that $r$
  has an outgoing $a$-transition, for each $u \in V$, we create an edge of
  probability~$1$ in~$H'$ from $(u,r)$ to~$(u,a)$. Further,
  for each edge $e$ of~$H$ labelled~$a$, writing $e = (u, v)$, we create
  the edge $e' = ((u,a),(v,r'))$ in~$H'$ with probability $\pi(e') \colonequals
  \pi(e)$, for $r'$ the unique target state of all $a$-transitions in the local
  DFA~$A$.
  Last, add a source $s$ in~$H'$ with edges of probability~$1$ to~$(u,r_0)$ for each $u
  \in V$ with $r_0$ being the initial state of~$A$, and a sink $t$ in~$H'$ with edges of probability~$1$ from~$(u,r)$ for each $u
  \in V$ and each final state $r$ of~$A$.

  The construction is in polynomial time, and the result $(H', \pi')$ is indeed a
  probabilistic graph: note that we do not create the same edge twice. Further,
  there is a clear probability-preserving bijection between the subgraphs
  of~$H$ and the subgraphs of~$H'$ that keep the additional edges of
  probability~$1$, which we call the \emph{possible worlds} of~$(H', \pi')$. We
  claim that $Q$ is satisfied in a possible world of
  $(H, \pi)$ iff there is an $st$-path in the corresponding possible world
  of~$(H', \pi')$.

  In the forward direction: if $Q$ is true in a possible world of~$(H,\pi)$, then this is witnessed by a sequence of edges
  $e_1, \ldots, e_n$ forming a word 
  $a_1 \cdots a_n$ that has an accepting run $r_0, \ldots, r_n$ in~$A$ in the
  DFA~$A$.
  The corresponding unlabelled edges in the
  corresponding possible world of~$H'$ must be present, which
  together with the edges of probability~1 gives us an $st$-path. More
  precisely, writing $e_i = (u_i,v_i)$ for each $1 \leq i \leq n$ with $v_i =
  u_{i+1}$ for each $1 \leq i < n$, the path
  is $s, (u_1, r_0), (u_1, a_1), \ldots, (u_n, r_{n-1}), (u_n, a_n), (v_n, r_n),
  t$. This uses the fact that $r_0$ is the initial state of~$A$, that each $r_{i-1}$
  has an outgoing $a_i$-transition for $1 \leq i \leq n$, that each $r_i$ is the
  target state of $a_i$-transitions for $1 \leq i \leq n$, and that $r_n$ is a
  final state of~$A$.

  In the backward direction: if there is an $st$-path in a possible world
  of~$H'$, then by construction of~$H'$ the path must be of the form $s, (u_1,
  r_0), (u_1, a_1), \ldots, (u_n, r_{n-1}), \allowbreak (u_n, a_n),
  \allowbreak (v_n, r_n), t$ for some
  choices of $u_1, \ldots, u_n, v_n \in V$ and $a_1, \ldots, a_n \in \Sigma$ and
  $r_0, \ldots, r_n$ states of~$A$, with $r_0, \ldots, r_n$ an accepting
  run of~$A$ on the word $a_1 \ldots a_n$. Further, we must have kept in the
  corresponding possible world of~$(H,\pi)$ the edges $e_1, \ldots, e_n$ with $e_i
  \coloneq (u_i,v_i)$ for each $1 \leq i \leq n$ and $v_i \coloneq u_{i+1}$ for
  each $1\leq i < n$, each $e_i$ being labeled $a_i$ for $1 \leq i \leq n$.
  Thus, this path witnesses that $Q$ is satisfied in that
  possible world of~$(H,\pi)$.
\end{proof}

So, all unbounded RPQs are at least as hard as ST-CON, and there are some
(namely, those whose infix-free language is local) that are polynomially
equivalent to ST-CON. We do not know if the class of local languages is 
maximal with respect to this property. However, we can construct examples of 
unbounded RPQs $Q$ whose language is not local, but for which PQE admits an 
FPRAS iff ST-CON admits an FPRAS. For instance:

\begin{prop}
  \label{prp:local}
  Let $L$ be an infix-free local language and let $L'$ be an infix-free finite
  language such that $L$ and $L'$ are on disjoint alphabets. Let $Q$ be the
  query corresponding to the disjunction of~$L$ and~$L'$, i.e., $\rpq(L 
  \sqcup L')$. Then $\pqe{Q}$ admits a FPRAS iff ST-CON admits an FPRAS.
\end{prop}

\begin{proof}
  For the forward direction, by Proposition~\ref{prp:stcon2rpq}, there is a
  PTIME reduction from ST-CON to $\pqe{Q}$: we can ensure that a match of $L'$
  is present by adding  for instance a disjoint path of edges with
  probability~$1$ forming a word of~$L'$, and this does not affect the 
  presence of a match of~$L$ because $L'$ and $L$ are on disjoint alphabets.

  For the backward direction, we know that $\pqe{\rpq(Q')}$ admits an FPRAS by
  Proposition~\ref{prp:finiterpqfpras}, and we know that $\pqe{\rpq(Q)}$
  reduces
  to ST-CON by Proposition~\ref{prp:rpq2stcon} which admits an FPRAS by 
  hypothesis. Notice that the probability of $L$ inducing a match on some 
  deterministic
  subgraph $H_d$ of the input probabilistic graph $H$ is independent of 
  the probability of $L'$ 
  inducing a match, since the alphabets of $L$ and $L'$ are disjoint. Let 
  $E_L$ denote the former event, and $E_{L'}$ the latter. Now, we wish to 
  compute an $(\epsilon,\delta)$-approximation of $\Pr(E_L \cup E_{L'})$ 
  (this is precisely the answer to $\pqe{\rpq(L \sqcup L')}$ on $H$). By the 
  inclusion-exclusion rule and independence of events, we have $\Pr(E_L 
  \cup E_{L'}) = \Pr(E_L) + \Pr(E_{L'}) - \Pr(E_L)\Pr(E_{L'})$. We may 
  approximately 
  evaluate this expression by calling the two FPRASes for $\Pr(E_L)$ and 
  $\Pr(E_{L'})$ respectively, both with error parameter $\tau = 
  \frac{\epsilon}{4}$ and confidence $\eta = \frac{\delta}{2}$. It is routine 
  to verify that this gives an $(\epsilon, \delta)$-approximation of $\Pr(E_L 
  \cup E_{L'})$ as desired.
\end{proof}

What about more expressive unbounded RPQs? Using techniques similar to
those in Claim~\ref{clm:2cnf2prov} earlier in the paper, we can show that some
unbounded RPQs \emph{do not} admit a FPRAS, conditionally to $\RP\neq\NP$.

\begin{prop}
  \label{prp:rpqhard}
  Let $Q$ be the unbounded RPQ defined from the infix-free language $a a (b^8 a)^* a$.
  Then $\pqe{Q}$
  does not admit an FPRAS in data complexity unless $\RP=\NP$.
\end{prop}

The proof of this result hinges on the following, which is an immediate
variant of Claim~\ref{clm:2cnf2prov}:

\begin{clm}
  \label{clm:rpqprov}
  Let $\Sigma$ be the alphabet $\{a, b\}$.
  Let $d > 1$ be a constant.
  Let $Q_d$ be the unbounded RPQ defined by the infix-free regular language $a a (b^{d+2} a)^* a$.
  Given a monotone 2-CNF formula $\phi$ on $n$ variables where each variable
  occurs in at most $d$ clauses, we can build in time $O(|\phi|)$ a graph
  $H_\phi$ in the class \all containing edges $(e_1, \ldots, e_n)$ such that
  $\prov^{Q_d}_{H_\phi}$ represents $\phi$
  on~$(e_1, \ldots, e_n)$.
\end{clm}

\begin{proof}
  We adapt the proof of Claim~\ref{clm:2cnf2prov}.
  Fix the constant $d > 1$.
  Define the RPQ $Q_d$ as above: notice the similarity with the
  \owp query of the proof of Claim~\ref{clm:2cnf2prov}, except that we replace
  the exponent $m$ with a Kleene star, and except that the labels $U$ and $R$
  from the signature in Claim~\ref{clm:2cnf2prov} respectively correspond to the
  letters $a$ and $b$ in the alphabet.

  Given the monotone 2-CNF formula $\phi$, we define the instance 
	graph $H_\phi$ in the class \all exactly like in the proof of Claim~\ref{clm:2cnf2prov}.

  We now observe the
  following, which allows us to conclude from the proof of
  Claim~\ref{clm:2cnf2prov}: for any subgraph $H_\phi'$
  of~$H_\phi$, we have that $H_\phi'$ satisfies the RPQ $Q_d$
  iff it satisfies the \owp query $G_\phi$ defined in the proof of
  Claim~\ref{clm:2cnf2prov}. One direction is immediate: if $H_\phi'$ satisfies
  the \owp query $G_\phi$, then it satisfies $Q_d$, because $G_\phi$ corresponds
  to the word $aa (b^{d+2} a)^m a$ of~$\Sigma^*$ which is a word of~$L$, so the
  image of a homomorphism of $G_\phi$ into $H_\phi'$ witnesses that $Q_d$ is
  satisfied.

  For the converse direction, assume that $H_\phi'$ satisfies the RPQ $Q_d$, and
  consider a \owp query $G_{\phi}'$ defining a word of $L$ such that
  $G_{\phi}'$ has a homomorphism into~$H_{\phi'}$. We claim that this word must
  be $a a (b^{d+2} a)^m a$, namely, the word that corresponds to~$G_\phi$, which
  suffices to conclude. Indeed, let us consider a match of such an \owp
  $G_{\phi}'$
  in~$H_\phi'$.
  The property ($\star\star$) of the proof of
  Claim~\ref{clm:2cnf2prov} implies that the $aa$ prefix of the match must be mapped to
  $\ledge{c_0'}{a}{} \ledge{c_0}{a}{d_0}$ and the $aa$ suffix of the
  match must be mapped 
          to $\ledge{c_m}{a}{} \ledge{d_m}{a}{d_m'}$. Now, it follows from
          the property ($\star$) of the proof of
  Claim~\ref{clm:2cnf2prov} that, in each factor of the form $a b^{d+2} a$, 
  the initial $a$ must be matched to an edge $\ledge{c_i}{a}{d_i}$ and the final
  $a$ must be matched to an edge $\ledge{c_j}{a}{d_j}$ such that $j = i+1$. It
  immediately follows that $G_{\phi}'$ must contain precisely $m+3$ edges
  labelled~$a$, which concludes the proof.
\end{proof}

Thanks to \cite[Theorem~2]{S10}, we can now show
Proposition~\ref{prp:rpqhard} from Claim~\ref{clm:rpqprov} using $d=6$,
in exactly the same way that 
Proposition~\ref{prop:phoml-owp-all} follows from Claim~\ref{clm:2cnf2prov}.

It is also easy to show an unconditional strongly exponential DNNF provenance 
lower bound on the same query as in Proposition~\ref{prp:rpqhard}, by 
adapting Proposition~\ref{prop:owp-all-lb}:

\begin{prop}\label{prop:rpq-lb}
        Let $Q$ be the RPQ from Proposition~\ref{prp:rpqhard}.
	There is an 
	infinite family $H_1, H_2, \ldots$ of labelled instances in the class
        \all such that, for any $i
	> 0$, any DNNF circuit representing the Boolean function
        $\prov^{Q}_{H_i}$ has
        size $2^{\Omega(\gsize{H_i})}$.
\end{prop}
\begin{proof}
        Like in the proof of Proposition~\ref{prop:owp-all-lb}, 
        we consider an infinite family $\phi_1, \phi_2, \ldots$ of monotone 2-CNF
	formulas where each variable occurs in at most $3$ clauses, such that
        any DNNF computing $\phi_i$ must have size $2^{\Omega(|\phi_i|)}$ for all~$i > 1$.
        Claim~\ref{clm:rpqprov} then gives us an infinite family 
        $H_1, H_2, \ldots$ of graphs in the class \all such that $\prov^{Q}_{H_i}$
        represents $\phi_i$ on some choice of edges, and we have $\gsize{H_i} =
	O(|\phi_i|)$ for all $i>0$ (from the running time bound). Now, 
        again, any representation of
	$\provl^Q_{H_i}$ as a DNNF can be translated in linear time to a
	representation of $\phi_i$ as a DNNF of the same size, simply by
        renaming variables and replacing variables by constants. This means that the lower bound on the size of DNNFs
	computing $\phi_i$ implies our desired lower bound.
\end{proof}

The results of
Proposition~\ref{prp:rpqhard} and
Proposition~\ref{prop:rpq-lb} are shown for a specific choice of RPQ, but
they do not apply to arbitrary RPQs. Indeed, by Proposition~\ref{prp:local},
there are some RPQs for which the existence of a FPRAS is equivalent to the
existence of a FPRAS for ST-CON, which is open.
Hence, a natural question left open by the present work is to characterize
the RPQs for which we can show
similar inapproximability results or DNNF provenance lower bound results. In particular, 
establishing DNNF provenance lower bounds for ST-CON is another interesting 
question left open here.

\section{Conclusions and Future Work}\label{sec:conclusions}
We studied the existence and non-existence of \textit{combined approximation 
algorithms} for the PQE problem, as well as the existence of 
polynomially-sized tractable circuit representations of provenance, under the 
lens of combined complexity. We additionally considered the approximability 
of regular path queries in data complexity.

We see several potential directions for future work. First, it would be interesting to see if the results in Proposition~\ref{prop:phoml-owp-dagi} and Theorem~\ref{thm:twoterminal-dag-fpras} can be extended beyond \dagi instances: graph classes of bounded \textit{DAG-width}~\cite{BDHKO12} could be a possible candidate here. We also leave open the problem of filling in the two remaining gaps in Table~\ref{tab:tab}. Namely, we would like to obtain either an FPRAS or hardness of approximation result for the equivalent problems $\phomul(\owp,\all)$ and $\phomul(\dwt,\all)$.
It is also natural to ask whether our results can be lifted from graph signatures to arbitrary relational signatures, or whether they apply in the \emph{unweighted} setting where all edges are required to have the same probability~\cite{amarilli2022uniform,amarilli2023uniform,kenig2021dichotomy}. Another question is whether we can classify the combined complexity of approximate PQE for \emph{disconnected} queries, as was done in~\cite{AMS17} in the case of exact computation, for queries that feature disjunction such as UCQs (already in the exact case~\cite{AMS17}), or for more general query classes, \eg, with recursion~\cite{amarilli2022dichotomy}. 

Lastly, for the data complexity of PQE for regular path queries, it remains 
open whether
our results on specific query classes can be generalized to a full dichotomy 
characterizing 
which
RPQs are approximable and which are (conditionally) not approximable; or which
RPQs admit tractable provenance representations as DNNFs (or subclasses
thereof) and which do not. However, 
such a result would require in particular to know whether ST-CON is approximable, 
and whether it admits tractable DNNF provenance representations.

\bibliographystyle{alphaurl}
\bibliography{article}

\appendix

\section{Proof of Theorem~\ref{thm:nobdd-weighted-count}}
\label{apx:nobdd-weighted-count}
\nobddweightedcount*
\begin{proof}
	We may assume without loss of generality that $D$ contains no variable $v$ such that $w(v) = 0$ or $ w(v) = 1$, since any such variable can be dealt with in constant time by conditioning $D$ accordingly.
By~\cite[Lemma~3.16]{ACMS20}, we may assume that $D$ is \textit{complete}, that is, every path from the root to a sink of $D$ tests every variable of $V$. 
        We will use the fact that for any positive integer $n$ and set of variables $S = \{x_1, \dots, x_k\}$ such that $k \geq \lceil \log n \rceil$, we can construct in time $O(k)$ a complete OBDD $C_{n}(x_1, \dots, x_k)$, implementing a ``comparator'' on the variables of $S$, that tests if the integer represented by the binary string $x_1 \dots x_{k}$ is strictly less than $n$ (hence, $C_{n}(x_1, \dots, x_k)$ has precisely $n$ satisfying assignments, for every sufficiently large value of $k$). In particular, when $k = 0$ then $C_1$ is the trivial OBDD comprising only a 1-sink.
	
        As $D$ is complete, there is a natural bijection between the models of the Boolean function captured by $D$, and the paths from the root to the 1-sink of $D$. Now, perform the following procedure for every variable label $v_i$ with weight $w(v_i) = p_i/q_i$ appearing in $D$. Set $k = \left \lceil{\log \left( d+1 \right)}\right \rceil$, where $d = \max\{p_i, q_i - p_i\}$. Send the $1$-edge emanating from every node $r \in D$ labelled with $v_i$ to the OBDD $C_{p_i}(x_1, \dots, x_{k})$ (where $x_1, \dots, x_k$ are fresh variables), redirecting edges to the $1$-sink of $C_{p}(x_1, \dots, x_{k})$ to the original destination of the $1$-edge from $n$. Do the same for $0$-edge from $r$, but with the OBDD $C_{q_i - p_i}(x_1, \dots, x_{k})$. Observe that $D$ remains a complete nOBDD. Moreover, it is not difficult to see that there are now exactly $p_i$ paths from the root to the 1-sink of $D$ that pass through the $1$-edge emanating from $r$, and $q_i - p_i$ paths passing through the $0$-edge.
	
	After repeating this process for every variable in $D$, we may apply Theorem~\ref{thm:nobdd-unweighted-count}, before normalizing the result by the product of the weight denominators $\prod q_i$.
\end{proof}
\section{Proof of Proposition~\ref{prop:phomul-owp-dagi}}
  \label{apx:phomul-owp-dagi}
\phomulowpdagi*
\begin{proof}
	As mentioned in the main text, the positive result directly follows from the existence of an FPRAS in the
	labelled setting, which was shown in
        Proposition~\ref{prop:phoml-owp-dagi}. It remains to show
        \SPhardness here. We will define a reduction from the
        \SPhard problem \countppdnf from
        \cite[Theorem~3.1]{DBLP:series/synthesis/2011Suciu}, which asks to
        count the satisfying valuations of a Boolean formula that is in
        \mbox{2-DNF} (i.e., in disjunctive normal form with two variables per
        clause), that is
        positive (i.e., has no negative literals), and that is partitioned (i.e.,
        variables are partitioned between two sets and each clause contains one
        variable from each set). Formally, \countppdnf asks us to count the
        satisfying valuations of an input formula of the form $\phi = \bigvee_{(i,j) \in E} \left( X_i \land Y_j \right)$. We reduce to $\phomul(G, \dagi)$, where $G$ is fixed to be the \owp of length three, \ie, $\rightarrow \rightarrow \rightarrow$.
	
	Construct a probabilistic graph $H = (\{s, t\} \sqcup X \sqcup Y, E_H)$ with probability labelling~$\pi$, where $s$ and $t$ are fresh vertices, $X = \{X_1, \dots, X_m \}$ and $Y = \{ Y_1, \dots, Y_n \}$ are vertices corresponding to the variables of $\phi$, and the edge set $E_H$ comprises:
	\begin{itemize}
		\item a directed edge $s \rightarrow X_i$ for every $X_i \in X$, with probability $0.5$;
		\item a directed edge $X_i \rightarrow Y_j$ for every $(i, j) \in E$, with probability
		$1$ (\ie, one directed edge per clause of $\phi$);
		\item a directed edge $Y_j \rightarrow t$ for every $Y_j \in Y$, with probability $0.5$.
	\end{itemize}
        It is clear that $H \in \dagi$. Moreover, we claim that $\Pr_\pi(G \homo H)$ is precisely the number of satisfying valuations of $\phi$, divided by $2^{|X \sqcup Y|}$. Indeed, just consider the natural bijection between the subgraphs of $H$ and the
	valuations of $\phi$, where we keep the edge $s \rightarrow X_i$ iff $X_i$ is assigned to true in a given valuation, and the edge $Y_j \rightarrow t$ iff $Y_j$ is assigned to true. Then it is easy to check that a subgraph of $H$ admits a homomorphism from $G$ iff the corresponding valuation satisfies $\phi$.
\end{proof}
\section{Proof of Proposition~\ref{prop:phomul-dwt-dagi}}
  \label{apx:phomul-dwt-dagi}
\phomuldwtdagi*
\begin{proof}
	Hardness follows directly from Proposition~\ref{prop:phomul-owp-dagi}, so we show the positive result here. Let $G$ be a \dwt query graph, and $m$ its height, \ie, the length of the longest directed path it contains. Let $G'$ be this \owp of length $m$, computable in polynomial time from $G$. We claim the following.
	\begin{clm}
		For any $H \in \dagi$, $\phomul(G, H) = \phomul(G', H)$.
	\end{clm}
	\begin{proof}
		Certainly, if $H' \subseteq H$ admits a homomorphism from $G$, then it admits one from $G'$ too since $G' \subseteq G$. On the other hand, if $H'$ admits a homomorphism from $G'$, then it also admits one from $G$: just map all vertices of distance $i$ from the root of $G$ to the image of the $i$-th vertex of $G'$.
	\end{proof}
	\noindent Now the result follows from the FPRAS for $\phomul(G', \dagi)$
	given by Proposition~\ref{prop:phoml-owp-dagi}.
\end{proof}

\end{document}